\documentclass[11pt]{article}
\def\submission{0}

\usepackage{amsmath,amssymb,amsfonts,amsthm,epsfig}
\usepackage[usenames,dvipsnames]{xcolor}
\usepackage{bm,xspace}
\usepackage{tcolorbox}
\usepackage{cancel}
\usepackage{fullpage}
\usepackage{guy}
\usepackage{framed}
\usepackage{verbatim}
\usepackage{enumitem}
\usepackage{array}
\usepackage{multirow}
\usepackage{afterpage}
\usepackage{mathrsfs}
\usepackage{pifont} 
\usepackage{chngpage}
\usepackage[normalem]{ulem}
\usepackage{boxedminipage}
\usepackage{caption}
\usepackage{forest}
\usepackage{svg}
\usepackage{microtype}
\usepackage{setspace}
\usepackage{xspace}

\usepackage{soul}

\usepackage{algorithm}
\usepackage{algorithmicx}
\usepackage[noend]{algpseudocode}

\usepackage{tikz}
\usepackage{pgfplots}
\pgfplotsset{compat=1.18}
\usepgfplotslibrary{groupplots}

\usepackage{dashbox}

\makeatletter
\newcommand{\subalign}[1]{%
  \vcenter{%
    \Let@ \restore@math@cr \default@tag
    \baselineskip\fontdimen10 \scriptfont\tw@
    \advance\baselineskip\fontdimen12 \scriptfont\tw@
    \lineskip\thr@@\fontdimen8 \scriptfont\thr@@
    \lineskiplimit\lineskip
    \ialign{\hfil$\m@th\scriptstyle##$&$\m@th\scriptstyle{}##$\hfil\crcr
      #1\crcr
    }%
  }%
}
\makeatother

\newcommand{\Poi}{\mathrm{Poi}}
\newcommand{\Ber}{\mathrm{Ber}}
\newcommand{\Bin}{\mathrm{Bin}}

\newcommand{\Accept}{\textsf{accept}\xspace}
\newcommand{\Reject}{\textsf{reject}\xspace}
\newcommand{\Plausible}{\textsf{plausible}\xspace}
\newcommand{\Split}{\mathrm{Split}}
\newcommand{\Perm}{\mathrm{Perm}}

\newcommand{\unif}{\mathbf{u}}

\newcommand{\muI}{\mu_I}
\newcommand{\estI}{\mathrm{est}_I}
\newcommand{\est}{\mathrm{est}}
\newcommand{\bestI}{\mathbf{est}_I}

\newlist{enumprop}{enumerate}{1} % set up a dedicated enumeration environment
\setlist[enumprop]{label=\arabic*.,ref=\theproposition.\arabic*}
\crefalias{enumpropi}{proposition}

\newcommand{\pparagraph}[1]{\bigskip \noindent {\bf {#1}}}

\providecommand{\email}[1]{\href{mailto:#1}{\nolinkurl{#1}\xspace}}

\begin{document}

\title{Instance-Optimal Uniformity Testing and Tracking}
\ifnum\submission=1
\author{Author(s) name(s) withheld}
\else
\author{Guy Blanc\thanks{Stanford University. Email: \email{gblanc@stanford.edu}.}
\and Cl\'ement L. Canonne\thanks{University of Sydney. Email: \email{clement.canonne@sydney.edu.au}.}
\and Erik Waingarten\thanks{University of Pennsylvania. Email: \email{ewaingar@seas.upenn.edu}.}
}
\fi

\date{\today}

\maketitle

\begin{abstract}
In the uniformity testing task, an algorithm is provided with samples from an unknown probability distribution over a (known) finite domain, and must decide whether it is the uniform distribution, or, alternatively, if its total variation distance from uniform exceeds some input distance parameter. This question has received a significant amount of interest and its complexity is, by now, fully settled. Yet, we argue that it fails to capture many scenarios of interest, and that its very definition as a gap problem in terms of a prespecified distance may lead to suboptimal performance. 

To address these shortcomings, we introduce the problem of \emph{uniformity tracking}, whereby an algorithm is required to detect deviations from uniformity (however they may manifest themselves) using as few samples as possible, and be competitive against an optimal algorithm knowing the distribution profile in hindsight. Our main contribution is a $\operatorname{polylog}(\operatorname{opt})$-competitive uniformity tracking algorithm. We obtain this result by leveraging new structural results on Poisson mixtures, which we believe to be of independent interest.
\end{abstract}

% \input{Abstract}

%  \thispagestyle{empty}
%  \newpage 

%  \thispagestyle{empty}
%  \setcounter{tocdepth}{2}
% \tableofcontents

%  \thispagestyle{empty}
%  \newpage 

 \setcounter{page}{1}

%%%%%%%%%%%%%%%%%%%%%%%%%%%%%%%%%%%%%%%%%%%%
\section{Introduction}
You are stationed at a telescope station, tasked with monitoring a large array of sensors: every second, one of the many cells of the array goes off, recording the impact of a seemingly random particle. In normal times, there is no pattern to these impacts: each of the $n$ cells is equally likely to be hit. What your team is hoping for is something different: that somewhere in the portion of the sky you are focusing the telescopes on, there is something interesting. The signal would show this as a deviation from uniformity, with the random impacts on the array cells showing a different statistical pattern. But access to such a state-of-the-art telescope is a valuable resource, and observation time is precious: \emph{how quickly will you be able to detect a non-uniform signal? How many seconds do you need?}

In this hypothetical scenario, you may think you are in luck: this is exactly the type of questions \emph{distribution testing}, a field of property testing first introduced two decades ago by Goldreich, Goldwasser, and Ron~\cite{GoldreichGR96} and systematically studied by Batu, Fortnow, Rubinfeld, Smith and White~\cite{BatuFRSW00}, is concerned with. In distribution testing, you are given a distance parameter $\eps > 0$, and observe independent samples from an unknown probability distribution $\bp$ over a known domain of size $n$: fixing a ``property'' of interest $\mcP$ (i.e., a subset of probability distributions you care about), your goal is to take as few samples as possible, and distinguish with high probability of success between (1)~$\bp$ has the property (i.e., $\bp \in \mcP$), and (2)~$\bp$ is ``far'' from having the property, i.e., is at distance at least $\eps$ from every $\bq\in \mcP$; where the notion of distance used is typically the total variation (TV) distance.\footnote{For more on property testing and distribution testing, we refer the reader to, e.g.,~\cite{Ron09,Goldreich17,BY22} and~\cite{Rubinfeld12,Canonne20,C22b}.} Taking $\mcP$ to be the property ``being the uniform distribution over $n$ elements,'' we obtain the task of \emph{uniformity testing}, thoroughly studied and by now fully understood~\cite{GR11,Paninski08,DiakonikolasGPP18,DiakonikolasGPP19}: $\Theta(\sqrt{n}/\eps^2)$ observations are necessary and sufficient to test whether the signals recorded by the telescope cells are truly uniform, or if their distribution is at TV distance at least $\eps$ from uniformly random. 

This is great news\dots{} except that this is \emph{not} the question you are trying to answer. To begin, you do not really have a ``distance parameter'' $\eps$ in mind, and you do not want one: your goal is to detect non-uniformity, not ``non-uniformity by at least some fixed amount.'' But more importantly, \emph{you do not care about total variation distance.} Or Hellinger distance, or $\ell_2$, or Kullback--Leibler divergence, or any of the many distance measures one could consider between probability distributions. What you need is to detect non-uniform signals with as few observations as possible, provably: any testing algorithm tailored to a specific distance measure may fail to achieve this goal for some probability distributions. That is, what you really seek is a \emph{distance-agnostic}, \emph{instance-optimal} testing algorithm which, given a sequence of samples from any probability distribution $\bp$, detects non-uniformity as early as possible. Formalizing what this means, and providing such an algorithm, is the focus of this paper.

\subsection{The Uniformity Tracking Setting}
    We now rigorously introduce the problem we study, before motivating our definitional choice and making a few crucial observations. Throughout, we fix a universe size $n \in \N$, and denote by $\unif$ the uniform distribution over $[n]$. We start by defining \emph{uniformity tracking algorithms}, the main algorithmic notion of this work. We then describe and motivate how we will measure its performance, which will be a competitive analysis framework; before providing a simple reduction to a slightly different version of the problem, which we term \emph{instance-optimal uniformity testing}. Finally, we show that many of the known uniformity testing algorithms, while optimal from a worst-case perspective, are not optimal in the competitive framework.

\begin{definition}[Uniformity Tracking]\label{def:track}
Fix the universe size $n \in \N$ and failure probability $\delta \in(0,1]$. A \emph{uniformity tracking algorithm} is a randomized algorithm with the following guarantees:
\begin{itemize}
\item \emph{\textbf{Input}}: There is an unknown distribution $\bp$ supported on $[n]$, and at each time $t \in \N$, the algorithm receives an independent draw $\bx_t \sim \bp$.
\item \emph{\textbf{Output}}: The algorithms runs continuously and outputs \Plausible or \Reject after each sample. Once the algorithm outputs \Reject, it terminates.
\end{itemize}
The algorithm runs continuously and should satisfy the following completeness and soundness guarantees. 
\begin{itemize}
\item \emph{\textbf{Completeness}}: If $\bp = \unif$, the probability the algorithm ever outputs \Reject is at most $\delta$; 
\item \emph{\textbf{Soundness}}: If $\bp\neq \unif$, the algorithm eventually outputs \Reject.
\end{itemize}
For any $\bp\neq \unif$, the \emph{sample complexity} $t(\bp) \in \R_{\geq 0}$ is a function of $\bp$ only, and is the expected number of samples before the algorithm outputs \Reject.
\end{definition}

As discussed earlier, unlike in traditional distribution testing, a defining feature of a uniformity tracking algorithm (Definition~\ref{def:track}) is that there is no notion of a proximity parameter (often denoted ``$\eps$'' in property testing). Hence, we do not need to specify with respect to which distance measure our algorithm ``tests'': total variation distance, Hellinger distance, \dots, etc., as the tracking algorithm should \emph{always} eventually detect non-uniformity of any non-uniform $\bp$. As the algorithm tracks the samples from $\bp$ and non-uniformity remains undetected, it is plausible that the unknown distribution $\bp$ is uniform; however, the moment the algorithm has convincing evidence that $\bp$ is non-uniform, it must reject as soon as possible. A simple consequence is that an algorithm never fully ``accepts'': for any finite $s \in \N$, there exists a distribution $\bp$ which is not uniform, but indistinguishable from uniform with fewer than $s$ samples---had the algorithm accepted after $s$ samples, it would inevitably err on some inputs. 

A crucial question is then how to quantify the performance of a uniformity tracking algorithm. In particular, if we require an algorithm to be competitive against ``the best possible algorithm,'' we must define a meaningful notion of what it means to be ``best possible'' without being vacuous. We introduce next the notion of competitiveness we will consider, before discussing and motivating it:
\begin{definition}[Competitiveness Against Relabelings]
\label{def:opt}
For a distribution $\bp$ supported on $[n]$, let $\mcP(\bp)$ be the class of all distributions obtained from $\bp$ after relabeling indices. We let
\[ 
\opt(\bp) = \min\left\{ t \in \N : \begin{array}{c} \exists \text{ algorithm which accepts $\unif$ and rejects all $\tilde{\bp} \in \mcP(\bp)$} \\ \text{after $t$ samples     w.p. at least $9/10$} \end{array}\right\}.\]
For $c \geq 1$, a tracking algorithm is \emph{$c$-competitive} if it satisfies $t(\bp) \leq c \cdot \opt(\bp)$ for all $\bp \neq \unif$, and we refer to $c$ its \emph{competitive ratio}.
\end{definition}
We emphasize that this choice to consider all relabelings of $\bp$ is far from arbitrary, but instead stems from the uniformity tracking task itself: it is natural to require algorithms to be competitive against the label-invariant class $\mcP(\bp)$, since the uniform distribution is itself label-invariant. % Being competitive with respect to $\opt(\bp)$ prevents the algorithm from focusing on any one particular labeling. (Guy: removing this one sentence because we have said ``labeling" so much in this one paragraph)
In this work, we design a single algorithm which is guaranteed to be $\polylog(\opt)$-competitive:%\gnote{Removed the dependence on $n$. It's just $\polylog(\opt)$}
\begin{theorem}[Main Result: Uniformity Tracking (Informal)]
     \label{thm:main-tracking}
     There is a efficient\footnote{Throughout, we use \emph{efficient} to indicate computational efficiency (polynomial running time in all parameters).} $\polylog(\opt(\bp), 1/\delta)$-competitive uniformity tracking algorithm.
\end{theorem}
This competitive ratio guarantees that, for \emph{every} non-uniform $\bp$, the number of samples taken is best possible (up to a poly-logarithmic factors) \emph{even among algorithms which are given as extra information the full ``profile'' of $\bp$}.\footnote{Recall that the profile of a probability distribution $\bp$ is, for each $\alpha \in [0, 1]$, the number of elements of $\bp$ whose probability is $\alpha$.} The guarantee can also be understood through the following lens: (1) the algorithm will not reject the uniform distribution (except with probability at most $\delta$), and any distribution $\bp$ which is non-uniform is eventually rejected after an expected $t(\bp)$ samples. Furthermore, (2) when it does reject a non-uniform distribution, the algorithm has a compelling ``excuse'' for not having terminated sooner: had it output \Reject in fewer than $t(\bp) / \polylog(t(\bp))$ samples, it would have been incorrect on some relabeling of $\bp$.

Via a straightforward reduction (\cref{lem:guess}), this will in turn follow from the following result on a variant of the problem, which we refer to as \emph{instance-optimal uniformity testing}:\footnote{This task is not to be confused with the  use of the term ``instance-optimal'' by Valiant and Valiant~\cite{ValiantV17}, which they introduced to refer to \emph{identity testing} (a generalization of uniformity testing) in the standard setting of distribution testing with respect to total variation distance, but with sample complexity parameterized by the reference distribution instead of the domain size. We discuss this, as well as the relation to the notion of competitive testing of~\cite{AcharyaDJOP11}, in~\cref{ssec:related}.}
\begin{theorem}[Instance-Optimal Uniformity Testing]
     \label{thm:main-instanceoptimal}
    There is an efficient algorithm which, given a sequence of $m$ samples from an unknown distribution $\bp$ supported on $[n]$ and parameter $\delta\in(0,1]$, has the following guarantees:
    \begin{enumerate}
        \item \emph{\textbf{Completeness}}: If $\bp= \unif$, it outputs $\Accept$ with probability at least $1-\delta$.
        \item \emph{\textbf{Soundness}}: If $\bp \neq \unif$ and $\opt(\bp) \leq m/\polylog(m,1/\delta)$, it outputs $\Reject$ with probability at least $1-\delta$.
    \end{enumerate}
\end{theorem}
\noindent That is, the algorithm rejects any non-uniform $\bp$ after taking only $\tilde{O}(\opt(\bp))$ samples.

\begin{remark}[Impossibility of a stronger goal.]%\gnote{Changed to remark so I could cite in technical overview}
    \label{remark:impossibility-stronger}
    A more ambitious goal would define $\opt'(\bp)$ to be the minimum $t$ such that an algorithm which accepts $\unif$ and rejects $\bp$ after $t$ samples exists. In other words, to require an algorithm to compete against one which knew $\bp$ itself, not just its profile. Unfortunately, a poly-logarithmic competitive ratio for this notion $\opt'(\bp)$ is easily seen to be unachievable, since there is a simple $\Omega(\sqrt{n})$ lower bound on the competitive ratio. Namely, let $\bp$ be a distribution which is uniform on a random $n/10$-sized subset of $[n]$: such distributions are known to require $\Omega(\sqrt{n})$ samples by a simple birthday paradox argument. On the other hand, $\dtv(\unif, \bp) \geq 9/10$, so $\opt'(\bp) = 1$. More generally, the natural hard examples for uniformity testing give families of distribution $\bp$ which individually satisfy $\dtv(\unif, \bp) \geq \eps$ and are distinguishable in $O(1/\eps^2)$ samples with knowledge of $\bp$, but no algorithm can distinguish a randomly generated $\bp$ from uniform (without knowledge of $\bp$) in fewer than $\Omega(\sqrt{n} / \eps^2)$ samples.
\end{remark}
% \paragraph{Impossibility of a stronger goal.} A more ambitious goal would define $\opt(\bp)$ to be the minimum $t$ such that an algorithm which accepts $\unif$ and rejects $\bp$ after $t$ samples exists. In other words, require an algorithm to compete against one which knew $\bp$ in hindsight. A poly-logarithmic competitive ratio for this notion of $\opt(\bp)$ is too strong, since there is a simple $\Omega(\sqrt{n})$ lower bound on the competitive ratio. Let $\bp$ be a distribution which is uniform on a random $n/10$-sized subset of $[n]$, such distributions are known to require $\Omega(\sqrt{n})$ samples (otherwise, there are no ``collisions'', and samples look uniform). On the other hand, $\dtv(\unif, p) \geq 9/10$, so $\opt(\bp) = 1$. More generally, the natural hard examples for uniformity testing give families of distribution $\bp$ which individually satisfy $\dtv(\unif, \bp) \geq \eps$ and are distinguishable in $O(1/\eps^2)$ samples with knowledge of $\bp$, but no algorithm can distinguish a randomly generated $\bp$ from uniform (without knowledge of $\bp$) in fewer than $\Omega(\sqrt{n} / \eps^2)$ samples.

\paragraph{Reduction to a Known Value of $\opt(\bp)$.} The first step is a standard reduction to an algorithm which outputs \Accept or \Reject, and which is provided with an upper bound on $\opt(\bp)$. The algorithm simply proceeds by repeated doubling of a budget, which only affects the competitive ratio by a constant factor.
\begin{lemma}\label{lem:guess}\
Fix $n \in \N$ and $\delta \in (0, 1)$. Suppose that there exists an algorithm which, for any $m \in \N$, receives as input $s(m) \leq c(m)\cdot m$ samples from an unknown distribution $\bp$ supported on $[n]$, and satisfies:
\begin{itemize}
\item \emph{\textbf{Completeness}}: If $\bp= \unif$, it outputs \Accept with probability at least $1-\delta / (2m)$.
\item \emph{\textbf{Soundness}}: If $\bp \neq \unif$ and $\opt(\bp) \leq m$, it outputs \Reject with probability at least $9/10$.
\end{itemize}
Then, there exists a $O(c)$-competitive uniformity tracking algorithm.
\end{lemma}
We provide the proof of this lemma in~\cref{appendix:proof-of-guess}. 
In what follows, and consistent to its use in~\cref{thm:main-instanceoptimal}, we will refer to an algorithm satisfying the assumptions of the above lemma as an \emph{instance-optimal uniformity testing algorithm}.

\paragraph{Worst-case optimal uniformity testers may not be competitive.} Consider the collision-based uniformity testing algorithm, which after taking $m$ samples, counts the fraction of pairwise collisions, and outputs \Accept or \Reject based on whether this fraction is higher or lower than a fixed threshold. This tester is known to be optimal in the classical setting where performance is measured with respect to $\dtv(\bp, \bu)$ \cite{DiakonikolasGPP19}. 

Given $m$ samples from the uniform distribution, the expected fraction of distinct pairs $(i, j)$ which ``collide'', i.e., where $\bx_i = \bx_j$, is exactly $\| \unif \|_2^2 = 1 / n$, since this is the probability that two draws from a uniform distribution are the same. Furthermore, the variance of this fraction is at most $4 / (m^2 n)$ (see page 26 in~\cite{C22b}). Consider the following non-uniform distribution $\bp$: the first $n-1$ elements appear with probability $(1-\beta)/n$ and the last element with probability $\beta + (1-\beta)/n$ (for some free parameter $\beta$). Here, the expected fraction of pairs which collide is $\|\bp\|_2^2 = \beta^2 + (1-\beta^2) / n$.
\begin{itemize} 
\item For this case, the number of samples in a collision-based test must be large enough so that the variance on the fraction of collisions in the uniform case is smaller than $\beta^4$ (so that there is a gap between the number of collisions observed in $\unif$ and $\bp$).\footnote{Specifically, we need to have $\| \unif \|_2^2 + \sqrt{4 / (m^2 n)} \ll \|\bp\|_2^2$, for the variance in the uniform case not to ``drown'' the difference in expectation between the two cases.} Doing so requires $4 / (m^2 n ) \leq \beta^4$, and so $m = \Omega(1/(\beta^2 \sqrt{n}))$.
\item However, there is a better test for $\bp$. Output $\Accept$ iff every element of $[n]$ appears $m /n \pm O(\sqrt{m \log n / n})$ times. By standard concentration results, the uniform distribution will have, once $m \gg n$, every element appearing $m / n \pm O(\sqrt{m \log n / n})$ times. In $\bp$, the last element appears $\beta m$ times in expectation, so it suffices to set $m = O(\log n / (\beta^2 n))$, a factor of $\tilde{\Omega}(\sqrt{n})$ better than the collision test. 
\end{itemize}

\paragraph{This is not just about not knowing ``$\eps$.''} As the above remark on worst-case uniformity testers makes it clear, the main conceptual feature of uniformity tracking is not simply the absence of a distance parameter $\eps$ as input: if this were the case, a simple doubling search for the ``right'' value $\eps=\dtv(\bp,\unif)$ would suffice, at the cost of only a doubly-logarithmic factor in $1/\eps$ (see e.g,~\cite{DK17,OufkirFFG21} or~\cite[Exercise 2.11]{C22b}). This doubling search technique could be applied to ``uniformity testing without knowing $\eps$'' for any prespecified notion of distance, not just total variation: yet, as the preceding discussion makes clear, any choice of the commonly studied distances for uniformity testing will fail to capture the instance-optimal aspect of the task.
\subsection{Related work}
\label{ssec:related}
There is a large (and growing) body of work in distribution testing: we here only cover the literature most relevant to our work, and refer the reader to~\cite[Chapter~11]{Goldreich17} and~\cite{Canonne20} for a broad overview of the field, and to~\cite{C22b} for a specific focus on uniformity and identity testing. For a statistician's take on these questions, see~\cite{BalakrishnanW18}.

Uniformity testing was first (and somewhat implicitly) considered in theoretical computer science in~\cite{GoldreichR00} (see also~\cite{GR11}); its optimal sample complexity of $O(\sqrt{n}/\eps^2)$ was then first claimed in~\cite{Paninski08},\footnote{As detailed in~\cite{C22b}, while the algorithm provided in~\cite{Paninski08} is itself correct, its analysis had a flaw.} before being re-established, with different algorithms and for the full range of distance parameters, in~\cite{AcharyaDK15,DiakonikolasGPP19}, and with the correct dependence on the error probability $\delta$ in~\cite{HuangM13,DiakonikolasGPP18}. A range of papers, starting with~\cite{BatuFRSW00} and including~\cite{AcharyaDK15,DiakonikolasK16}, then considered \emph{identity} testing, a generalization of uniformity testing where the reference distribution need not be uniform, but any fixed known $\bq$. \cite{DiakonikolasK16}, followed by~\cite{Goldreich20}, showed that the two tasks of uniformity and identity testing (with respect to total variation distance, and parameterized by the domain size $n$) are formally equivalent, as they can efficiently be reduced to each other.

An influential paper of Valiant and Valiant~\cite{ValiantV17} then introduced a notion of \emph{instance-optimal} identity testing, arguing that, along with the total variation distance parameter, $\eps$, one could (and should) parameterize the sample complexity by a function of the reference distribution $\bq$ instead of the domain size $n$. Followup work~\cite{DiakonikolasK16,BlaisCG19} continued in this direction, and provided alternative characterizations of the sample complexity \emph{via} a different function of the reference $\bq$. This notion of instance-optimality is further discussed in~\cite{ValiantV20}: we note that, following Goldreich, we instead refer to this setting of identity testing as \emph{massively parameterized identity testing}. In particular, the notion of instance-optimality we use in our work is \emph{not} the one introduced by Valiant and Valiant: indeed, we focus solely on the uniform distribution as reference distribution (i.e., uniformity testing), and request instance-optimality of the sample complexity \emph{with respect to the unknown distribution $\bp$ being tested.}

The above papers focus on distribution testing with respect to total variation distance, sometimes using other measures of distance as a proxy (e.g., $\ell_2$ distance in~\cite{ChanDVV14,DiakonikolasK16} or $\chi^2$ divergence in~\cite{AcharyaDK15}). \cite{Waggoner15}, and then more systematically \cite{DKW18}, focused on distribution testing with respect to alternative notions of distance for their own sake: for instance, \cite{DKW18} provides a sample-optimal identity (and thus uniformity) testing algorithm with respect to Hellinger distance. The thrust of their results, while conceptually different from ours, does follow a similar theme~--~namely, that testing with respect to a single notion of distance does not fully capture what one may need, and that testing with respect to other distances may lead to vastly different worst-case instances.

\paragraph{Competitive closeness testing.} We conclude this related work section by discussing the two papers most relevant to ours: that of Acharya, Das, Jafarpour, Orlitsky, and Pan~\cite{AcharyaDJOP11}, and its follow-up by Acharya, Das, Jafarpour, Orlitsky, Pan, and Suresh~\cite{AcharyaDJOPS12}. Their setting, which they term \emph{competitive closeness testing}, is very similar to our notion of instance-optimal uniformity testing. In theirs, an algorithm is provided two sequence of samples, from two unknown distributions $\bp$ and $\bq$: the task is to decide whether $\bp$ and $\bq$ are equal or different (and be correct with high probability). As in our case, in~\cite{AcharyaDJOP11,AcharyaDJOPS12} there is no notion of distance involved; moreover, they require the algorithms to be \emph{symmetric} (i.e., invariant under joint relabeling of the domain of the two sequences), and seek to be competitive against the best symmetric algorithm equipped with full knowledge of $\bp,\bq$~--~or, equivalently since the algorithms are symmetric, with knowledge of the profiles of $\bp,\bq$. However, while the two settings are very similar on a technical level, the conceptual aspects and motivation are quite different: in their work, the label invariance comes from the choice to study a specific class of algorithms, while in ours it stems directly from the symmetry of the problem (i.e., testing the uniform distribution).

The part of their work most relevant to us is their lower bound, which implies that departing from the uniform distribution to allow arbitrary reference $\bq$ would rule out any subpolynomial competitive ratio, let alone polylogarithmic (or constant). Specifically, for every $m \in \N$, they show the existence of a reference distribution $\bq$ and a family of distributions $\mathcal{F}$ (with $\bq\notin\mathcal{F}$) such that, given samples from any $\bp\in\mathcal{F}$  (1) knowing the profile of $\bp$, there is a testing algorithm distinguishing $\bp$ from $\bq$ with $m$ samples; but (2) without knowing that profile, any algorithm needs $\tilde{\Omega}(m^{7/6}) = \tilde{\Omega}(m^{1/6})\cdot m$: that is, no competitive ratio better than $\tilde{\Omega}(\opt^{1/6})$ can be achieved. Contrast this to our results, which show that when $\bq=\unif$ one can achieve $\polylog(\opt)$-competitiveness.

%Finally, on a technical level, their results on closeness testing establish the impossibility of any subpolynomial competitive ratio, let alone polylogarithmic (or constant): interestingly, their lower bound carries over to the ``competitive identity testing'' task (although they do not consider it in their work), showing that no instance-optimal identity testing (in our sense) algorithm can achieve good competitive ratio. Specifically, for every $m \in \N$, they show the existence of a reference distribution $\bq$ and a family of distributions $\mathcal{F}$ (with $\bq\notin\mathcal{F}$) such that, given samples from any $\bp\in\mathcal{F}$  (1) knowing the profile of $\bp$, there is a testing algorithm distinguishing $\bp$ from $\bq$ with $m$ samples; but (2) without knowing that profile, any algorithm needs $\tilde{\Omega}(m^{7/6}) = \tilde{\Omega}(m^{1/6})\cdot m$: that is, no competitive ratio better than $\tilde{\Omega}(\opt^{1/6})$ can be achieved. Contrast this to our results, which show that when $\bq=\unif$, then one can achieve $\polylog(\opt)$-competitiveness.
 
%\cmargin{Discuss changepoint detection?}

\subsection{Future work: Beyond uniformity}
This work shows that uniformity tracking and testing, can be performed with sample size competitive with the optimal samples needed for the unknown distribution $\bp$, where $\opt(\bp)$ is measured with respect to all relabelings of $\bp$. This naturally prompts a number of followup distribution testing problems. For what properties $P$ can we obtain an algorithm that determines whether $\bp \in P$ with sample complexity scaling nearly linearly with $\opt_P(\bp)$, and what is the right notion of $\opt_P(\bp)$ for more general properties? In general, it seems appropriate that the notion of $\opt_P(\cdot)$ should depend on the symmetry of the property $P$, as ours does for the uniform distribution. 

We list two concrete directions to explore. The first is to consider properties that are still fully invariant under relabelings of the distributions (the so-called \emph{symmetric properties}), but are broader than the property of ``being uniform." For example, is there an instance-optimal tester for the property of having high entropy, where $\opt$ is still defined with respect to all relabelings? To the best of our knowledge, it is conceivable that all convex (i.e. properties closed under taking mixtures) symmetric properties admit a nearly instance-optimal tester.

The second direction is to consider properties that are not fully invariant under relabelings of the distribution. For example, if the domain is $\mcX \coloneqq \zo^n$, is there an instance-optimal tester that determines whether an unknown distribution $\bp$ over $\mcX$ is uniform over some subspace of $\mcX$? Here, part of the challenge is designing the right notion of $\opt$. Intuitively, that definition should only allow relabelings that respect the linear structure of $\mcX$. Does such an $\opt$ suffice?

%%%%%%%%%%%%%%%%%%%%%%%%%%%%%%%%%%%%%%%%%%%%

\section{Technical Overview}
\label{sec:technical-overview}

A tester receives as input a list of frequencies $\bx_1, \ldots, \bx_n$, each corresponding to the number of times it sampled some particular element $i \in [n]$. We wish for our tester to be universal in the sense that, given $m' = \tilde{O}(m)$ samples from some $\bp$, if there is any test that distinguishes $\bp$ and all its relabelings from $\bu$ using $m$ samples, then our test should work using $m'$ samples without needing to know $\bp$. %A natural first attempt is to union bound over the failure probability of all possible tests. Since a test can be thought of as a function $T: [m]^n \to \zo$, a naive count gives that there are $2^{m^n}$ possible tests, far too many to directly run. Our tester can be seen as discretization this large set down to only polynomially many tests, a doubly exponential improvement.
A natural first attempt would be to (i)~consider all possible tests $T \colon [m]^n \to \{0,1\}$ which do not often reject the uniform distribution, and (ii)~run all these tests in parallel, while increasing the success probability so as to union bound over the probability that any test incorrectly rejects the uniform distribution. Alas, the number of such tests is too large (namely, $2^{m^n}$), which means that this algorithm is inefficient, and the overhead incurred from the union bound is $O(m^n)$; yielding only a $O(m^n)$-competitive algorithm. One way to understand our algorithm is as providing a randomized discretization of this large set down to only polynomially many tests, a doubly exponential improvement.

The first step is the well-known ``Poissonization" trick to make the frequencies independent. This allows us to assume, without loss of generality, that we receive independent frequencies, with the $i^{\text{th}}$ frequency drawn from $\Poi(\bp(i) \cdot m)$. Unfortunately, it will not suffice for our purposes to analyze fully independent frequencies. Instead, our notion of $\opt$-competitiveness requires us to understand distributions that are formed by first drawing $n$ independent frequencies and then uniformly permuting them, which results in intricate correlations between frequencies.
\begin{definition}[Permutation distribution]
    \label{def:perm-dist}
    For any $\bp_1, \ldots, \bp_n$, the distribution $\Perm(\bp_1, \ldots, \bp_n)$ is defined as the distribution of $\bx_1, \ldots, \bx_n$ obtained by the following process:
    \begin{enumerate}
        \item Draw a uniform permutation of $[n]$, $(\bi_1, \ldots, \bi_n)$;
        \item For each $j \in [n]$, independently sample $\bx_i \sim \bp_{\bi_j}$.
    \end{enumerate}
\end{definition}
We will require our tester to succeed whenever $\Perm(\Poi(m \cdot \bp(1)), \ldots, \Poi(m \cdot \bp(1)))$ is ``distinguishable" from the frequencies corresponding to the uniform distribution, which are simply $\Poi(m/n)^n$.
\begin{definition}[Distinguishability]
    \label{def:distinguishable}
    We say that distributions $\bp$ and $\bq$ are \emph{distinguishable} if $\dtv(\bp, \bq) \geq \Omega(1)$. Similarly, we say that $\bp$ and $\bq$ are \emph{distinguishable with $m$-samples} if $\bp^m$ and $\bq^m$ are distinguishable.
\end{definition}
With this language, we can state the main focus of this technical overview.
\begin{restatable}[Reformulation of~\cref{thm:main-instanceoptimal} using Poissonization]{theorem}{MainPoisson}
     \label{thm:main-poisson}
    For any $\mu \geq 0$ and $s = \polylog(n, \mu, 1/\delta)$, there is a efficient tester $T$ such that,
    \begin{enumerate}
        \item[$\circ$] Given a sample from $\Poi(s\cdot \mu)^n$, $T$ outputs $\Accept$ with probability at least $1 - \delta$.
        \item[$\circ$] For any $\lambda_1, \ldots, \lambda_n$ for which  $\Perm(\Poi(\lambda_1), \ldots, \Poi(\lambda_n))$ is distinguishable from $\Poi(\mu)^n$, given a sample from $\Perm(\Poi(s\lambda_1), \ldots, \Poi(s\lambda_n))$, $T$ outputs $\Reject$ with probability at least $1 - \delta$.
    \end{enumerate}
\end{restatable}
The parameter $s$ in \Cref{thm:main-poisson} corresponds to the sample complexity overhead in~\cref{thm:main-instanceoptimal}: If we had $s = O(1)$, then the tester of~\cref{thm:main-instanceoptimal} would reject given samples from any $\bp$ for which $\opt(\bp)$ is within a constant factor of $m$. We note that if we had required \Cref{thm:main-poisson} to reject whenever $\Poi(\lambda_1) \times \cdots \times \Poi(\lambda_n)$, rather than its permuted variant, were distinguishable from $\Poi(\mu)^n$, then a polynomial loss ($s \geq \poly(n)$) would be required for the reasons discussed in \Cref{remark:impossibility-stronger}.
% We note that a simpler version of \Cref{thm:main-poisson} would require $s \geq \poly(n)$: If it were required to reject whenever $\Poi(\lambda_1) \times \cdots \times \Poi(\lambda_n)$ is distinguishable from $\Poi(\mu)^n$, then such a polynomial loss would be required for the same reasons discussed in \Cref{remark:impossibility-stronger}. Hence, although Poissonization does impose some structure in the frequency distribution, we are forced to analyze dependent frequencies.
\paragraph{The structure of the proof.} The proof of \Cref{thm:main-poisson} has two main steps.
\begin{enumerate}
    \item \textbf{Subsampling to reduce to independent mixtures:} The first step is a randomized reduction which allows us to move from dependent to independent frequencies. We show it suffices to design a tester that distinguishes a single known Poisson from an unknown mixture of Poissons with nearly optimal sample complexity. That is, the new tester must reject if given $m' \coloneqq \tilde{O}(m)$ i.i.d.\ samples from a Poisson mixture $\bp = \frac{\Poi(\lambda_1) + \cdots + \Poi(\lambda_k)}{k}$ that is distinguishable from $\Poi(\mu)$ using $m$ samples. This simplifies our task as we no longer need to reason about the correlation structure of the frequencies.
    
    To do so, we show the following more general result: For any distributions $\bq$ and $\bp_1, \ldots, \bp_n$ for which $\bq^n$ and $\Perm(\bp_1, \ldots, \bp_n)$ are distinguishable, there is some subsampling size $k \in [n]$ for which, with ``moderate" probability over indices $\bi_1, \ldots, \bi_k$ chosen uniformly without replacement from $[n]$, the mixture distribution $\frac{\bp_{\bi_1} + \cdots + \bp_{\bi_k}}{k}$ is ``moderately" distinguishable from $\bq$. We note that this step does not necessarily hold \emph{without subsampling} (i.e., if we fixed $k = n$): One can easily construct distributions $\bp_1, \ldots, \bp_n$ whose mixture is exactly $\bq$, and yet, $\Perm(\bp_1, \ldots, \bp_n)$ is easily distinguishable from $\bq$. Intuitively, in these cases, the distributions $\bp_1, \ldots, \bp_n$ must be quite far apart, which is why one might expect a subsampled mixture to be far from $\bq$.
    \item \textbf{The interval tester for mixtures of Poissons:} We show for any single Poisson $\bq \coloneqq \Poi(\mu)$ and mixture of Poissons $\bp \coloneqq \frac{\Poi(\lambda_1) + \cdots + \Poi(\lambda_k)}{k}$, if $\bp$ and $\bq$ are distinguishable using $m$ samples, there is a simple ``interval tester" that distinguishes $\bp$ and $\bq$ using $\tilde{O}(m)$ samples. This interval tester furthermore depends only on $\bq$ and need not know the parameters $\lambda_1, \ldots, \lambda_k$. We formally give its pseudocode in  \Cref{fig:interval tester}. Briefly, it enumerates over polynomially many intervals. For each interval, it counts how many of its samples appeared in that interval and rejects iff any of these counts deviate significantly from their expectation under $\bq$.

    % \gray{In proving the correctness of this tester, we prove a more general result that may be of independent interest. We show that for any distributions $\bp, \bq$ distinguishable using $m$ samples, there is a simple ``counting test" that distinguishes them using $O(m \log m)$ samples: For a set $S_{\bp, \bq}$, this test counts how many of its $O(m \log m)$ samples are in $S_{\bp, \bq}$ and then and decides whether to accept or reject purely as a function of this count. To prove correctness of the interval tester in our setting, we further show the set $S_{\bp, \bq}$ is always an interval when $\bp$ is a single Poisson and $\bq$ is a mixture of Poissons.}
    %that there is a single ``counting test" which, for any distributions $\bp, \bq$ distinguishable using $m$ samples, distinguishes $\bp$ from $\bq$ using $O(m \log m)$ samples, with an extremely simple form: For some set $S$ that depends on $\bp$ and $\bq$, this test counts how many of its $O(m \log m)$ samples are in $S$ and then and decides whether to accept or reject purely as a function of this count. T
\end{enumerate}

\subsection{The interval tester for mixtures of Poissons}
\label{subsec:warm-up-overview}
We begin by sketching a warm-up result: For any single Poisson $\bq \coloneqq \Poi(\mu)$ and mixture of Poissons $\bp \coloneqq \frac{\Poi(\lambda_1) + \cdots + \Poi(\lambda_k)}{k}$, we show that if $\bp$ and $\bq$ are distinguishable using $m$ samples, then the aforementioned interval test works with $m' \coloneqq \tilde{O}(m^2)$ samples, a quantitatively weaker but simpler result than the $m' \coloneqq \tilde{O}(m)$ we ultimately show. For this, we use a standard property of total variation distance, its subadditivity for product distributions:
\begin{equation*}
    \dtv(\bp^m, \bq^m) \leq m \cdot \dtv(\bp, \bq).
\end{equation*}
Hence, our assumption that $\bp$ and $\bq$ are distinguishable in $m$ samples, which means $\dtv(\bp^m, \bq^m) \geq \Omega(1)$, implies that $\dtv(\bp, \bq) \geq \Omega(1/m)$. By the definition of total variation distance, this means for the set $S$ consisting of all points on which $\bq$ has more probability mass than $\bp$ (the so-called ``Scheff\'e set'' of $\bq$ and $\bp$), we have that
\begin{equation*}
   \Prx_{\bx \sim \bq}[\bx \in S] \geq \Prx_{\bx \sim \bp}[\bx \in S] + \Omega(1/m).
\end{equation*}
% As a result, in the setting of \Cref{lem:iid-nearly-opt}, we are interested in mixtures  $\bp \coloneqq \Ex_{\bi \in [n]}[\Poi(\lambda_{\bi})]$ for which $\dtv(\Poi(\mu), \bp) \geq \Omega(1/m)$. By definition, this means for the set $S$ consisting of all points on which $\Poi(\mu)$ has more probability mass than $\bp$ (the so-called ``Scheff\'e set'' of $\Poi(\mu)$ and $\bp$), we have that
% \begin{equation*}
%    \Prx_{\bx \sim \Poi(\mu)}[\bx \in S] \geq \Prx_{\bx \sim \bp}[\bx \in S] + \Omega(1/m).
% \end{equation*}
Suppose we knew this set $S$. Then, a simple test for distinguishing between $\bp$ and $\bq$ using $m^2$ samples is as follows:
\begin{enumerate}
    \item Count how many of the $m^2$ samples falls within $S$.
    \item Output $\Accept$ iff this count is in $[qm^2 \pm \Theta(m)]$ where $q \coloneqq \Prx_{\bx \sim \bq}[\bx \in S]$.
\end{enumerate}
The difficulty is that the set $S$ may depend on the distribution $\bp$ which in turn depend on the parameters $\lambda_1, \ldots, \lambda_n$, and the tester must be agnostic to these parameters. To circumvent this issue, we leverage the structure of Poisson distributions to show that the set $S$ must have a simple form: namely, that it is always an interval.
\begin{proposition}
    \label{prop:range-TV}
    For any $\mu$ and $\lambda_1, \ldots, \lambda_n$, the set
    \begin{equation*}
        S \coloneqq \set*{x \in \N, \Poi(\mu)(x) \geq \frac{\Poi(\lambda_1)(x) + \cdots + \Poi(\lambda_k)(x)}{k}}
    \end{equation*}
    satisfies $S = [a,b]$ or\footnote{The case where $\overline{S}$ is an interval occurs if $S =[a, \infty)$ for some integer $a$ in which case $\overline{S} = [0, a-1]$} $\overline{S} = [a,b]$ for some integers $a,b$.
\end{proposition}
While there are infinitely many intervals, we show, using standard concentration inequalities, that it suffices to consider a number of intervals that is polynomial in the problem parameters. This allows us to design an instance-optimal tester as follows. For every potential interval $S$, we run the aforementioned tester that assumed we knew $S$. While doing so, we increase the sample size by a factor logarithmic in the number of distinct intervals tested to allow for a union bound over the tests.

% Due to \Cref{prop:range-TV}, we can design a universal tester as follows. 

% First, we enumerate over all

% $a,b \in [u]$ where $u$ is a upper bound on the magnitude of points likely to be sampled by $\Poi(\mu)$. Then, we run the aforementioned tester with $S = [a,b]$. It's not too hard to show that if we increase the sample size by a factor logarithmic in the number of distinct $S$ we try, then this procedure will work. Since the number of distinct $S$ is only polynomial in the problem parameters, this overhead is logarithmic.

The proof of \Cref{prop:range-TV} rests on a simple observation: the ratio of probability mass functions (PMFs) of any two Poisson distributions is convex.
\begin{restatable}[The ratio of Poisson PMFs is convex]{proposition}{convexratio}
    \label{prop:ratio-single-convex}
    For any $\lambda_1, \lambda_2 \geq 0$, the function $x \mapsto \Poi(\lambda_1)(x)/\Poi(\lambda_2)(x)$ is convex.
\end{restatable}

Since the PMF of the mixture $\bp \coloneqq \frac{\Poi(\lambda_1) + \cdots + \Poi(\lambda_k)}{k}$ is a nonnegative linear combination of the PMFs of Poissons, we obtain the following as an immediate corollary:
\begin{restatable}{corollary}{convexmixture}
     \label{cor:ratio-mixture-convex}
    For any $\bq \coloneqq \Poi(\mu)$ and $\bp \coloneqq \frac{\Poi(\lambda_1) + \cdots + \Poi(\lambda_k)}{k}$, the function $x \mapsto \bp(x)/\bq(x)$ is convex.
\end{restatable}

\Cref{cor:ratio-mixture-convex} suffices to prove \Cref{prop:range-TV}. See \Cref{fig:poisson-ratio} for a depiction of these convex ratios and why \Cref{cor:ratio-mixture-convex} implies \Cref{prop:range-TV}.
\begin{figure}[ht!]
\centering
\begin{tikzpicture}

\begin{groupplot}[
        group style={group size=2 by 1, horizontal sep=1cm},
        width=9cm, height=7cm, tickpos=left
    ]

    \nextgroupplot[
        ybar,
        ymin=0,
        enlargelimits=0.05,
        xlabel={$x$},
        xtick={2,4,6,8,10,12,14,16,18,20}, 
        xmin=2, xmax=20,
        legend pos=north east,
        yticklabels={}, 
        bar width=7pt,
        title={Probability mass functions}
    ]

    \addplot[blue, fill=blue, opacity=0.4, bar shift=-1pt] coordinates {
        (2,0.002270)
(3,0.007567)
(4,0.018917)
(5,0.037833)
(6,0.063055)
(7,0.090079)
(8,0.112599)
(9,0.125110)
(10,0.125110)
(11,0.113736)
(12,0.094780)
(13,0.072908)
(14,0.052077)
(15,0.034718)
(16,0.021699)
(17,0.012764)
(18,0.007091)
(19,0.003732)
(20,0.001866)
    };
    \addlegendentry{$\boldsymbol{q} =\mathrm{Poi}(10)$}

    \addplot[red, fill=red, opacity=0.4, bar shift=1pt] coordinates {
        (2,0.042129)
(3,0.070273)
(4,0.088056)
(5,0.088702)
(6,0.075531)
(7,0.057408)
(8,0.042361)
(9,0.034336)
(10,0.033372)
(11,0.037265)
(12,0.043147)
(13,0.048464)
(14,0.051454)
(15,0.051297)
(16,0.048041)
(17,0.042375)
(18,0.035308)
(19,0.027874)
(20,0.020905)
    };
    \addlegendentry{$\bp =\frac{\Poi(5) + \Poi(15)}{2}$}

\nextgroupplot[
        ymin=0,
        enlargelimits=0.05,
        xlabel={$x$},
        xtick={2,4,6,8,10,12,14,16,18,20},
        xmin=2, xmax=20,
        title=Ratio of probability mass functions,
        legend pos=north east
    ]

    % Ratio of PMFs
    \addplot[violet, thick, mark=o] coordinates {
(1,37.10834)
(2,18.55923)
(3,9.28719)
(4,4.65497)
(5,2.34454)
(6,1.19785)
(7,0.63730)
(8,0.37621)
(9,0.27445)
(10,0.26674)
(11,0.32764)
(12,0.45523)
(13,0.66473)
(14,0.98803)
(15,1.47752)
(16,2.21401)
(17,3.31989)
(18,4.97926)
(19,7.46861)
(20,11.20277)
(21,16.80409)
(22,25.20610)
    };
    \addlegendentry{$\bp(k)/\bq(k)$}

    \addplot[dashed, thick] coordinates {(-5,1) (25,1)};
    \addlegendentry{Baseline ($\text{ratio}=1$)}
\end{groupplot}
\end{tikzpicture}
\caption{The left plot shows two probability mass functions: One corresponding to $\bq \coloneqq \Poi(10)$ and the other a mixture $\bp \coloneqq \frac{\Poi(5) + \Poi(15)}{2}$. The right plot shows the ratio $\bp(x) / \bq(x)$, which is convex by \Cref{cor:ratio-mixture-convex}. \Cref{prop:range-TV} corresponds to the observation that this ratio is less than $1$ for an interval.}
\label{fig:poisson-ratio}
\end{figure}
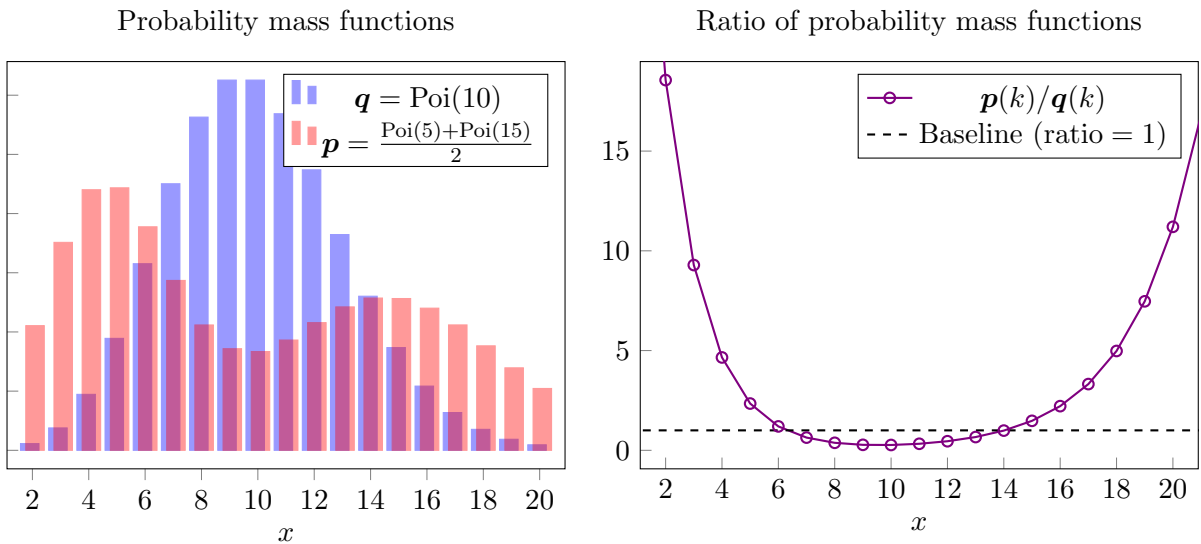

\pparagraph{Improving the sample complexity.}
Using a recent result of Pensia, Jog, and Loh, we are able to show that essentially the same interval tester succeeds using only $\tilde{O}(m)$ samples rather than $O(m^2)$ samples.
\begin{theorem}[Theorem 3.2, Corollary 3.4 of \cite{PJL23}]
    \label{thm:PJL-main}
    For any distribution $\bp$ and $\bq$ that are distinguishable using $m$ samples, there exists a test distinguishing $\bp$ and $\bq$ using $m' \coloneqq O(m \log m)$ samples, of the following simple form: For some $r, \tau$ which depend on $\bp$ and $\bq$, the tester takes in samples $x_1, \ldots, x_{m'}$ and outputs $\Accept$ iff $\bp(x_i)/\bq(x_i) \geq r$ for at least $\tau$ choices of $i\in [m']$.
\end{theorem}

The motivation of \cite{PJL23}'s result was to communication-constrained hypothesis testing, where the fact that their test needs only remember a single bit about each sample is key. They furthermore showed that this single bit can always be computed by thresholding the value of $p(x)/q(x)$ for computational complexity reasons: If the test designer knows $\bp$ and $\bq$, they can brute force the values of $r$ and $\tau$ efficiently to design the communication-constrained test. In our setting, this particular structure has additional utility: Even \emph{without knowing} $\bp$, as long as it is a mixture of Poissons, such ``threshold tests" always reduce to counting the number of elements in some interval.

\begin{proposition}[Generalization of \Cref{prop:range-TV}]
    \label{prop:threshold-to-interval}
    For any $\bq \coloneqq \Poi(\mu)$, $\bp \coloneqq \frac{\Poi(\lambda_1) + \cdots + \Poi(\lambda_k)}{k}$, and $r \geq 0$, there exists integers $a,b$ for which
    \begin{equation*}
        S \coloneqq \set*{x \in \N,\, \frac{\bp(x)}{\bq(x)} \geq r}
    \end{equation*}
    either satisfies $S = [a,b]$ or $\overline{S} = [a,b]$.
\end{proposition}
\noindent As a consequence, we are able to show the following.
\begin{lemma}[An instance optimal tester for mixtures of Poissons]
    \label{lem:iid-nearly-opt}
    For any $\mu,\delta \geq 0$ and $m \in \N$, there is an efficient test $T$ (see \Cref{fig:interval tester}) using $m' \coloneqq m\cdot \polylog(m, \mu, 1/\delta)$ samples such that,
    \begin{enumerate}
        \item[$\circ$] Given $m'$ independent samples from $\Poi(\mu)$, $T$ outputs $\Accept$ with probability at least $1 - \delta$;  
        \item[$\circ$] For any $\lambda_1, \ldots, \lambda_k$ for which $\Poi(\mu)$ and $\bp \coloneqq \frac{\Poi(\lambda_1) + \cdots + \Poi(\lambda_k)}{k}$ are distinguishable using $m$ samples, given $m'$ independent samples from $\bp$, $T$ outputs $\Reject$ with probability at least $1 - \delta$.
    \end{enumerate}
\end{lemma}

\subsection{Subsampling to reduce to independent mixtures}

Recall that we aim to prove \Cref{thm:main-poisson} where, in the $\Reject$ case, the frequencies come from a distribution of the form $\Perm(\Poi(\lambda_1), \ldots, \Poi(\lambda_n))$. Such a distribution is difficult to analyze because the $n$ frequencies are \emph{dependent}. We use the following to reduce this to the case where we have $n$ frequencies drawn \emph{independently} from a Poisson mixture, the setting where \Cref{lem:iid-nearly-opt} applies.

% \begin{proposition}[Distinguishing permutation distributions using subsamples]
%     \label{lem:subsample-is-distinguishable}   
%     For any distributions $\bp_1, \ldots, \bp_n$ and $\bq$ for which $\Perm(\bp_1, \ldots, \bp_n)$ is distinguishable from $\bq^n$, there is some $k \in [n]$ for which
%      \begin{equation*}
%         \Ex_{\bi_1, \ldots, \bi_k \wor [n]}\bracket*{\dhels\paren* {\frac{\bp_{\bi_1} + \cdots + \bp_{\bi_k}}{k}, \Poi(\mu)}} \geq \Omega(1/(k \log n)),
%     \end{equation*}
%     where $\bi_1, \ldots, \bi_k \wor [n]$ indicates that $\bi_1, \ldots, \bi_n$ are drawn uniformly without replacement from $[n]$.
% \end{proposition}

    \begin{lemma}[Distinguishing permutation distributions using subsamples]
    \label{lem:subsample-is-distinguishable}
    For any distributions $\bp_1, \ldots, \bp_n$ and $\bq$ for which $\Perm(\bp_1, \ldots, \bp_n)$ is distinguishable from $\bq^n$, there exist some $k \in [n]$ for which the following holds: With probability at least $\Omega(1/(k \log n))$ over $\bi_1, \ldots, \bi_k$ chosen uniformly without replacement from $[n]$ the mixture distribution $\frac{\bp_{\bi_1} + \cdots + \bp_{\bi_k}}{k}$ is distinguishable from $\bq$ using $O(k \log n)$ samples.
\end{lemma}

% \gray{
% \begin{restatable}[Distinguishing permutation distributions using subsamples]{lemma}{subsamplePermutationRemove}
%     %\label{lem:subsample-is-distinguishable}
%     For any distributions $\bp_1, \ldots, \bp_n$ and $\bq$ for which $\Perm(\bp_1, \ldots, \bp_n)$ is distinguishable from $\bq^n$, there exist some $k \in [n]$ for which,
%     \begin{equation*}
%         \Ex_{\bi_1, \ldots, \bi_{k} \wor [n]}\bracket*{\dhels\paren*{\bq, \frac{\bp_{\bi_1} + \cdots + \bp_{\bi_k}}{k}}} \geq \Omega\paren*{\frac{1}{k \log n}}.
%     \end{equation*}
%     where $\wor$ denotes sampling uniformly without replacement.
% \end{restatable}
% }

It is easy to construct counterexamples showing that \Cref{lem:subsample-is-distinguishable} would not hold \emph{without} subsampling (i.e., setting $k = n$). For example, let each $\bp_i$ be the distribution that outputs $i$ with probability $1$, and $\bq$ be uniform over $[n]$. In this case, we have that $\Perm(\bp_1, \ldots, \bp_n)$ is distinguishable from $\bq$, since a sample from $\Perm(\bp_1, \ldots, \bp_n)$ is guaranteed to have each element of $[n]$ exactly once, whereas for a sample from $\bq$ that is unlikely. On the other hand, the mixture $\frac{\bp_1 + \cdots + \bp_n}{n}$ is exactly equal to $\bq$. \Cref{lem:subsample-is-distinguishable} formalizes the intuition that, in such cases, the distributions $\bp_1, \ldots, \bp_n$ must themselves be far apart, which allows a random subsampled mixture to be distinguishable from $\bq$.

Our proof of \Cref{lem:subsample-is-distinguishable} relies on Hellinger distance (\cref{def:hel}), which is known to characterize, up to constant factors, the number of samples needed to distinguish two fixed distributions:
\begin{fact}[Hellinger distance characterizes distinguishability]
    \label{fact:hellinger=distinguish-remove}
    For any distributions $\bp$ and $\bq$,
    \begin{equation*}
        \min_{m \in \N}\paren[\big]{\dtv(\bp^m, \bq^m) \geq \Omega(1)}  = \Theta\paren*{\frac{1}{\dhels(\bp,\bq)}}.
    \end{equation*}
\end{fact}
It turns out, after restating \Cref{lem:subsample-is-distinguishable} in terms of Hellinger distance (see \Cref{lem:subsample-is-distinguishable-hellinger}), it follows fairly straightforwardly from a recent and surprisingly non-obvious result (see \Cref{fact:approx-chain-rule}) showing that squared Hellinger distance satisfies an approximate chain rule \cite{jay09, FHQR24}. 

Given \Cref{lem:subsample-is-distinguishable}, our tester can enumerate all values $k \in [n]$ to find the good subsampling size. This gives the following test meeting the requirements of \Cref{thm:main-poisson}: on input $\mu\geq 0$ and $n\in \N$,
\begin{enumerate}
    \item For each $k \in [n]$ repeat the following $\poly(n)$ times:
    \begin{enumerate}
        \item Sample $\bi_1, \ldots, \bi_k \wor [n]$.
        \item Run the test of \Cref{lem:iid-nearly-opt} to distinguish $\Poi(\mu)$ from $\frac{\Poi(\lambda_{\bi_1}) + \cdots + \Poi(\lambda_{\bi_k})}{k}$. If it outputs $\Reject$, this test outputs $\Reject$.
    \end{enumerate}
    \item If we still have not outputted $\Reject$, output $\Accept$.
\end{enumerate}

A straightforward analysis gives that, if this test is given as input, for each $i \in [n]$, $s = \polylog(n, \mu)$ many samples from $\Poi(\lambda_i)$, it will succeed with high probability. While we do not have at our disposal this many samples of $\Poi(\lambda_i)$, the following standard ``trick'' enables to obtain them as long as we have, for each $i \in [n]$, \emph{one} sample from $\Poi(s \cdot \lambda_i)$.
\begin{fact}[Poisson splitting]
    For any $s \in \N$, there exists an efficient procedure that takes as input a single sample from $\Poi(s \lambda)$ (where $\lambda$ is unknown) and outputs $s$ i.i.d.\ draws from $\Poi(\lambda)$.
\end{fact}

\section{Preliminaries}
Throughout, we denote by $\Poi(\lambda)$ the Poisson distribution with parameter $\lambda \geq 0$, by $\Ber(p)$ the Bernoulli distribution with parameter $p \in[0,1]$, by $\Bin(n,p)$ the Binomial distribution with parameters $n\in\N$ and $p\in[0,1]$, and $\bu$ the uniform distribution over $[n]$ (with $n$ typically clear from context). Other than these standard distributions, we will use $\bp, \bq$ and their various subscripts to denote other distributions. Random variables will also be written in \textbf{boldfont} and use symbols other than $\bp$ or $\bq$, (e.g. $\bx \sim \bp$).

For any distributions $\bp_1 ,\ldots, \bp_n$, we use $\frac{\bp_1 + \cdots + \bp_n}{n}$ to denote the uniform mixture over $\bp_1, \ldots, \bp_n$. In particular, it will often be convenient to refer to mixtures of Poisson distributions without specifying the parameters of the mixture.
\begin{definition}[Poisson mixture]
    \label{def:poisson-mixture}
    We say a distribution $\bp$ is a \emph{Poisson mixture} if there is some $n \in \N$ and $\lambda_1, \ldots, \lambda_n$ for which $\bp = \frac{\Poi(\lambda_1) + \cdots + \Poi(\lambda_n)}{n}$.
\end{definition}

We recall the definitions of some notions of distances between probability distributions that our proofs will rely on. For simplicity, we focus on the discrete, finite domain case, as this will be sufficient for our purposes.
\begin{definition}[Total variation distance]
    \label{def:tv}
    For any two distributions $\bp, \bq$ on a discrete domain $\mcX$, the \emph{total variation distance} (or statistical distance) $\dtv(\bp,\bq)$ between $\bp$ and $\bq$ is defined by
    \begin{equation}
        \dtv(\bp,\bq) =  \sup_{S\subseteq \mcX} (\bp(S) - \bq(S)) = \frac{1}{2}\sum_{x\in\mcX} |\bp(x)-\bq(x)|.
    \end{equation}
    In particular, $\dtv$ is a metric, and bounded in $[0,1]$.
\end{definition}

% \begin{definition}[Distinguishable distributions]
%     \label{def:distinguishable}
%     We say distributions $\bp$ and $\bq$ \emph{are distinguishable using $m$ samples} if $\dtv(\bp^m, \bq^m) \geq \Omega(1)$. For the special case of $m = 1$, we simply say $\bp$ and $\bq$ \emph{are distinguishable.}
% \end{definition}

\begin{definition}[Hellinger distance]
    \label{def:hel}
    For any two distributions $\bp, \bq$ on a discrete domain $\mcX$, the \emph{Hellinger distance} $\dhel(\bp,\bq)$ between $\bp$ and $\bq$ is defined by
    \begin{equation}
        \dhels(\bp,\bq) =  \sum_{x \in \mcX}\paren*{\sqrt{\bp(x)} -\sqrt{\bq(x)}}^2
    \end{equation}
    In particular, $\dhel$ is a metric, and bounded in $[0,2]$.\footnote{The Hellinger distance is often defined with a factor $1/\sqrt{2}$ to make it take values in $[0,1]$. To ease notation, and as it is not necessary in our setting, we omit this normalization factor.}
\end{definition}

\begin{definition}[Kullback--Leibler divergence]
    \label{def:kl}
    For any two distributions $\bp, \bq$ on a discrete domain $\mcX$, the \emph{Kullback--Leibler (KL) divergence} (or relative entropy) $\KL{\bp}{\bq}$ between $\bp$ and $\bq$ is defined by
    \begin{equation}
        \KL{\bp}{\bq} =  \sum_{x \in \mcX} \bp(x) \ln \frac{\bp(x)}{\bq(x)}\,,
    \end{equation}
    where $\ln$ is the natural logarithm, and by convention $0\ln 0 = 0$. 
    Note that the KL divergence is non-negative and unbounded, but not symmetric, and does not satify the triangle inequality.
\end{definition}

Total variation distance can be related to Hellinger distance through the following standard inequalities (see, e.g.,~\cite[Lemma~B.1]{C22b}):
\begin{equation}
    \label{eq:tv:hellinger}
    \frac{1}{2}\dhels(\bp,\bq) \leq \dtv(\bp,\bq) \leq \dhel(\bp,\bq)
\end{equation}
for every distributions $\bp, \bq$.\smallskip

% \gray{
% A prominent object in our work is the \emph{permutation distribution}, defined below.
% \begin{definition}[Permutation distribution]
%     %\label{def:perm-dist}
%     For any $\bp_1, \ldots, \bp_n$, the distribution $\Perm(\bp_1, \ldots, \bp_n)$ is defined as the distribution of $\bx_1, \ldots, \bx_n$ obtained by the following process:
%     \begin{enumerate}
%         \item Draw a uniform permutation of $[n]$, $(\bi_1, \ldots, \bi_n)$;
%         \item For each $j \in [n]$, independently sample $\bx_i \sim \bp_{\bi_j}$.
%     \end{enumerate}
% \end{definition}
% Of particular interest to us will be the case where $\bp_1, \ldots, \bp_n$ are all Poisson distributions, where $\bp_i = \Poi(\lambda_i)$. Another related notion in this setting is is that of a uniform mixture of Poissons. That is, for a vector of parameters $\Lambda = (\lambda_1,\dots, \lambda_n) \in \R_{\geq 0}^n$, we denote by $\mcM(\Lambda)$ the uniform mixture\gnote{It may be worth getting rid of this notation}
% \[
% \mcM(\Lambda) \coloneqq \frac{\Poi(\lambda_1) + \cdots + \Poi(\lambda_n)}{n}
% = \Ex_{\bi \sim [n]}[\Poi(\lambda_{\bi})]\,.
% \]}

\noindent We will use the following form of Chernoff bounds.
\begin{fact}[Chernoff Bounds]
    \label{fact:Chernoff-KL}
    Let $\bx_1, \ldots, \bx_n$ be independent random variables bounded in $[0,1]$ and $\overline{\bx} = \frac{1}{n} \sum_{i \in [n]} \bx_i$ their empirical mean. For $p \coloneqq \Ex[\overline{\bx}]$ and any $q \in [0,1]$,
    \begin{align*}
        \Prx[\overline{\bx} \geq q] \leq e^{-n \cdot \KL{\Ber(q)}{\Ber(p)}} \quad\quad\text{ if }q \geq p, \\
        \Prx[\overline{\bx} \leq q] \leq e^{-n \cdot \KL{\Ber(q)}{\Ber(p)}} \quad\quad\text{ if }q \leq p.
    \end{align*}
\end{fact}
We will use two corollaries of the above bound. First, using the fact that $\KL{\bp}{\bq} \geq \dhels(\bp, \bq)$ and union bounding over the above two inequalities, we have the following.
\begin{corollary}[Chernoff bounds: Hellinger distance version]
    \label{fact:Chernoff-hellinger}
    In the setting of \Cref{fact:Chernoff-KL},
    \begin{equation*}
        \Prx\bracket*{\dhels(\Ber(p), \Ber(\overline{\bx})) \geq \tau}\leq 2e^{-n \tau}.
    \end{equation*}
\end{corollary}
Then, using that $\KL{\Ber(q)}{\Ber(p)} \geq q/6$ whenever $q \geq 2p$, we have the following.
\begin{corollary}[Chernoff bounds in the tail]
    \label{fact:Chernoff-tail}
    In the setting of \Cref{fact:Chernoff-KL}, for any $\delta \in(0,1]$,
    \begin{align*}
        \Prx\bracket*{n\overline{\bx} \geq 2pn + \frac{\ln(1/\delta)}{6}} \leq \delta.
    \end{align*}
\end{corollary}
From standard arguments (applying the above bound to $\Bin(n,\mu/n)$, and taking the limit as $n\to \infty$), one obtains the analogous statement for Poisson random variables:
\begin{fact}[Tail bound for Poisson random variables]
    \label{fact:Poisson-tail}
    For any $\mu \geq 0$ and $\delta \in(0,1]$,
    \begin{equation*}
        \Prx_{\bz \sim \Poi(\mu)}\bracket*{\bz \geq 2\mu+ \frac{\ln(1/\delta)}{6}} \leq \delta.
    \end{equation*}
\end{fact}
\section{An instance-optimal tester for mixtures of Poissons: Proof of \Cref{lem:iid-nearly-opt}}
\label{sec:mixtures}
In this section, we show the interval tester distinguishes a single known Poisson distribution from any (unknown) mixture of Poisson distributions (recall \Cref{def:poisson-mixture}).
%We prove the formal version of \Cref{lem:iid-nearly-opt}, stated below:
\begin{lemma}[The interval test is instance-optimal to distinguish a single Poisson from a mixture, formal version of \Cref{lem:iid-nearly-opt}]
    \label{lem:iid-nearly-opt-body}
    For any $\mu, \eps,\delta \geq 0$, define
    \begin{equation*}
        x_{\max} \coloneqq \Theta(\mu + \log(1/\eps)), \quad\tau \coloneqq \Theta\paren*{\frac{\eps}{\log(4/\eps)}},\quad \text{and}\quad m \coloneqq \Theta\paren*{\frac{\log(x_{\max}/\delta)}{\tau}},
    \end{equation*}
    the interval tester of \Cref{fig:interval tester} with parameters $\mu, \tau,$ and $x_{\max}$, if given i.i.d.\ samples $\bx_1, \ldots ,\bx_m \sim \bp^m$ from an unknown distribution $\bp$ has the following guarantees:
    \begin{enumerate}
        \item[$\circ$] If $\bp = \Poi(\mu)$ it outputs $\Accept$ with probability at least $1-\delta$.
        \item[$\circ$] If $\bp$ is any Poisson mixture such that $\dhels(\Poi(\mu), \bp) \geq \eps$, it outputs $\Reject$ with probability at least $1 - \delta$.
    \end{enumerate}
\end{lemma}

\begin{algorithm}[htb] 

  %\captionsetup{width=.9\linewidth}

    \begin{algorithmic}[1]
    
    \Require Parameters $\mu, \tau, x_{\max}$ and samples $x_1, \ldots, x_m \in \N$.
    
    \ForAll{interval $I = [a,b]$ where $a \leq b \in [x_{\max}]$}
        \State Define the parameters
            \begin{equation*}
                \muI \coloneqq \Prx_{\bx \sim \Poi(\mu)}[\bx \in I]\quad\quad \estI\coloneqq \frac{1}{m} \sum_{i \in [m]} \Ind[x_i \in I].
            \end{equation*}
        \If{$\dhels(\Ber(\mu_I), \Ber(\est_I)) \geq \tau$
       } \Return $\Reject$ and halt.
        \EndIf
    \EndFor
    \State \Return $\Accept$
    \end{algorithmic}
\caption{The interval tester.}
\label{fig:interval tester}
\end{algorithm}

\noindent The remainder of this section is structured as follows:

\begin{enumerate}
    \item In \Cref{subsec:mixture-structure} we establish a structural result stating that if a single Poisson is distinguishable from a mixture of Poissons, there is an interval witnessing this distinguishability.
    \item In \Cref{subsec:iid-correctness}, we apply that structural result to prove \Cref{lem:iid-nearly-opt-body}.
\end{enumerate}

\subsection{A structural result for distinguishing a single Poisson from a mixture}
\label{subsec:mixture-structure}

We prove the following.
\begin{lemma}[Existence of a good interval for distinguishing Poissons]
    \label{lem:good-interval}
    For any $\bq= \Poi(\mu)$ and Poisson mixture $\bp$ satisfying
    \begin{equation*}
        \dhels(\bp, \bq) \geq \eps,
    \end{equation*}
    define,
    \begin{equation*}
         x_{\max} \coloneqq O(\mu + \log(1/\eps)).
    \end{equation*}
    Then an interval of the form $I = [a,b]$ for integers $0 \leq a \leq b \leq x_{\max}$ satisfies
    \begin{equation*}
       \dhels\paren*{\Ber\paren*{\bp(I)},\Ber\paren*{\bq(I)}}  \geq \Omega\paren*{\frac{\eps}{\log(4 /\eps)}}.
    \end{equation*}
\end{lemma}

As discussed in \Cref{sec:technical-overview}, the high-level idea behind \Cref{lem:good-interval} is to show that the function $x \mapsto \frac{\bp(x)}{\Poi(\mu)(x)}$ is convex. We then combine this with the below restatement of \Cref{thm:PJL-main}.
\begin{theorem}[Restatement of \Cref{thm:PJL-main} in terms of squared Hellinger distance, Corollary 3.4 of \cite{PJL23}]
    \label{thm:PJL-restated}
    For any distributions $\bp,  \bq$ with $\dhels(\bp,\bq) \geq \eps$ there exists some $r \geq 0$, so that, for $S \coloneqq \set{x, \bp(x)/\bq(x) \geq r}$,
    \begin{equation*}
        \dhels\paren*{\Ber\paren*{\bp(S)},\Ber\paren*{\bq(S)}} \geq \Omega\paren*{\frac{\eps}{\log(4/\eps)}}.
    \end{equation*}
\end{theorem}

% \gray{

% This, combined with \Cref{lem:counting-tests-near-optimal-body}, immediately gives a result similar to \Cref{lem:good-interval}, albeit weaker in two ways:
% \begin{enumerate}
%     \item Since a convex function $f$ can have up to two solutions to the equation $f(x) = y$, we may need two intervals rather than the one interval suggested by \Cref{lem:good-interval}. This aspect is not essential to obtain an instance-optimal tester, but would lead to a more complicated ``two-interval" tester.
%     \item We would not obtain any bound on $x_{\max}$. Without such a bound, our analysis would need to union bound over a potentially infinite number of intervals.
% \end{enumerate}
% Our analysis will ultimately alleviate these two issues to establish \Cref{lem:good-interval} in its stated form. Let us begin by proving the PMF ratio is convex.
% }
Next, we show the PMF ratio is convex.
\convexratio*
\begin{proof}
    This ratio is
    \begin{equation*}
        \frac{e^{-\lambda_1}\lambda_1^x}{x!} \cdot \frac{x!}{e^{-\lambda_2}\lambda_2^x} = e^{\lambda_2 - \lambda_1} \cdot (\lambda_1/\lambda_2)^x.
    \end{equation*}
    Since $x \mapsto c^x$ is convex for any $c \geq 0$ and $e^{\lambda_2 - \lambda_1}$ is a positive constant, the desired function is also convex.
\end{proof}
\begin{corollary}[Restatement of \Cref{cor:ratio-mixture-convex}]
    \label{cor:ratio-mixture-convex-body}
    For any $\bq \coloneqq \Poi(\mu)$ and Poisson mixture $\bp$, the function $x \mapsto \bp(x)/\bq(x)$ is convex.
\end{corollary}
\begin{proof}
    This follows from the fact that the PMF of a Poisson mixture is a nonnegative linear combination of Poisson PMFs and \Cref{prop:ratio-single-convex}.
\end{proof}

A direct combination of \Cref{cor:ratio-mixture-convex-body} and \Cref{thm:PJL-restated} would give a version of \Cref{lem:good-interval} without the constraint that we only need to search for intervals for which both end points are less than $x_{\max}$. For that, we'll use the following technical proposition.
\begin{restatable}{proposition}{eliminateLarge}
  \label{prop:eliminate-large}
    For any distributions $\bp$ and $\bq$, set $S$ for which $\dhels(\Ber(\bp(S)), \Ber(\bq(S))) \geq \delta$, and set $T$ for which $\bq(T) \leq \delta/20$, for either $S' \coloneqq S \setminus T$ or $S' \coloneqq \overline{T}$,
    \begin{equation*}
        \dhels(\Ber(\bp(S')), \Ber(\bq(S'))) \geq \delta/120.
    \end{equation*}
\end{restatable}
We defer the proof of \Cref{prop:eliminate-large} to \Cref{appendix:proof-of-eliminate-large}. Here, we prove the main result of this subsection.

\begin{proof}[Proof of \Cref{lem:good-interval}]
    Let $S$ be the set guaranteed to exist by \Cref{thm:PJL-restated}, which satisfies
    \begin{equation*}
        \paren*{\sqrt{\bp(S)} - \sqrt{\bq(S)}}^2 \geq \Omega\paren*{\frac{\eps}{\log(4/\eps)}}.
    \end{equation*}
    By the structural guarantee of \Cref{thm:PJL-restated}, we know this $S$ consists on points $x$ for which $f(x) \coloneqq \bp(x)/\bq(x)$ is at least some value $r$. Since $f$ is convex (by \Cref{cor:ratio-mixture-convex-body}), it can cross the line $g(x) =r$ at most twice (where a crossing is defined as $f(x) < r$ and $f(x+1) \geq r$ or $f(x) \geq r$ and $f(x+1) < r$ ). We show this implies that either $S = [a,b]$ or $\overline{S} = [a,b]$ for nonnegative integers $a,b$ by separating into cases based on the number of crossings.
    \begin{enumerate}
        \item[] \textbf{Two crossings:} If there are two crossings of $f(x)$ and the line $g(x) = r$ that occur at $a$ and $b$, then $\overline{S} = [a,b]$, as desired.
        \item[] \textbf{One crossings:} If there is one crossing occurring at $a$, then either $S = [0,a]$ or $\overline{S} = [0,a]$ as desired.
        \item[] \textbf{No crossings:} The only way zero crossings can occur is if $\bp = \bq$, in which case \Cref{lem:good-interval} is vacuously true as $\eps = 0$. This follows by observing that if zero crossings occur $S = \N$ or $S = \varnothing$, which means, either $\bp(S) = \bq(S) = 1$ or $\bp(S) = \bq(S) = 0$. By \Cref{thm:PJL-restated}, this implies $\eps = 0$.
    \end{enumerate}

    Note that by symmetry, $\dhels(\Ber(p), \Ber(q)) = \dhels(\Ber(1-p), \Ber(1-q))$. In particular, this implies that $\dhels\paren*{\Ber\paren*{\bp(S)},\Ber\paren*{\bq(S)}} = \dhels\paren*{\Ber\paren*{\bp(\overline{S})},\Ber\paren*{\bq(\overline{S})}}$. As a result, we have identified an $S$ of the form $[a,b]$ for integers $a \leq b$ for which $$\dhels\paren*{\Ber\paren*{\bp(S)},\Ber\paren*{\bq(S)}} \geq \delta \coloneqq \Omega\paren*{\frac{\eps}{\log(4/\eps)}}$$.

    If $b \leq x_{\max}$ we are done. Otherwise, by the Chernoff tail bound of \Cref{fact:Poisson-tail} and our choice of $x_{\max}$ we have that for $T = [x_{\max} +1, \infty)$, $\bq(T) \leq \delta/20$. Hence, by \Cref{prop:eliminate-large}, for either $S' \coloneqq S \setminus T$ or $S' \coloneqq \overline{T}$ we have that
    \begin{equation*}
        \dhels(\Ber(\bp(S')), \Ber(\bq(S'))) \geq \delta/120 \geq \Omega\paren*{\frac{\eps}{\log(4/\eps)}}.
    \end{equation*}
    In both cases, $S'$ will be an interval with both end points upper bounded by $x_{\max}$, as desired.

\end{proof}

\subsection{Correctness of the interval tester for mixtures, proof of \Cref{lem:iid-nearly-opt-body}}
\label{subsec:iid-correctness}
We conclude this section by applying \Cref{lem:good-interval} to prove \Cref{lem:iid-nearly-opt-body}. This analysis will use the ``almost triangle inequality" for squared $\ell_2$ distances, which we first recall.
\begin{fact}[Almost triangle inequality, special case of \Cref{prop:almost-triangle-gen} with $n=2$]
    \label{fact:almost-triangle-inequality}
    For any $a,b,c \in \R$,
    $
        (a-c)^2 \leq 2(a-b)^2 + 2(b-c)^2.
    $
\end{fact}
We will use the following corollary.
\begin{corollary}
    \label{cor:almost-triangle-hellinger}
    For any $\mu, p, \est,\tau \geq 0$ if
    \begin{equation*}
        \dhels(\Ber(\mu), \Ber(p)) \geq 4\tau \quad\quad\text{and}\quad\quad \dhels(\Ber(p), \Ber(\est)) \leq \tau
    \end{equation*}
    then
    \begin{equation*}
        \dhels(\mu, \Ber(\est)) \geq \tau
    \end{equation*}
\end{corollary}
\begin{proof}
    For any $p,q$, we can expand $\dhels(\Ber(p), \Ber(q)) = (\sqrt{p} - \sqrt{q})^2 + (\sqrt{1-p} - \sqrt{1-q})^2$. Applying \Cref{fact:almost-triangle-inequality} to each of the two terms, we have that
    \begin{equation*}
        \dhels(\Ber(\mu), \Ber(p)) \leq 2 \cdot \dhels(\Ber(p), \Ber(\est)) + 2  \cdot \dhels(\mu, \Ber(\est)).
    \end{equation*}
    Hence, $\dhels(\Ber(p), \Ber(\est)) \leq \tau$ and $\dhels(\mu, \Ber(\est)) < \tau$ would be a contradiction, implying the desired result.
\end{proof}

\begin{proof}[Proof of \Cref{lem:iid-nearly-opt-body}]
    We prove correctness for both cases.
    \begin{description}
    \item[Case one:] If $\bp = \Poi(\mu)$, we argue the test outputs $\Accept$ with probability at least $1-\delta$. In this case, for each interval $I$, we have that $\Ex[\bestI] = \muI$. Furthermore, $\bestI$ is the average of $m$ independent and bounded random variables. Hence, for any such interval, by \Cref{fact:Chernoff-hellinger},
    \begin{equation*}
        \Prx\bracket*{\dhels(\Ber(\muI), \Ber(\bestI)) \geq \tau} \leq 2e^{-m\tau}.
    \end{equation*}
    For a suitable choice of constants in the setting of $m$, this quantity is at most $O(\delta/ x_{\max}^2)$. Therefore, union bounding over the $O(x_{\max}^2)$ many intervals gives an upper bound of $\delta$ on the failure probability, as desired.

    \item[Case two:] If $\bp$ is a Poisson mixture and satisfies $\dhels(\Poi(\mu), \bp) \geq \eps$,  we argue the test outputs $\Reject$ with probability at least $1 - \delta$. In this case, by \Cref{lem:good-interval}, we know there is some interval $I$ that the tester of \Cref{fig:interval tester} checks for which,
    \begin{equation}
        \label{eq:good-interval}
        \dhels(\Ber(\muI), \Ber(\bp(I))%\gray{\paren*{\sqrt{\muI} - \sqrt{\bp(I)}}^2} 
        \geq  4\tau.
    \end{equation}
    By \Cref{cor:almost-triangle-hellinger},  if $\dhels(\Ber(\bestI), \Ber(\bp(I))\leq \tau$ then $\dhels(\Ber(\muI), \Ber(\bestI)) \geq \tau$ and the test outputs $\Reject$ as desired.    
    %For this interval, let $p = \Poi(\mu)(I)$ and $\bq$ be the random variable computed by the tester. By \Cref{fact:almost-triangle-inequality}, if $(\sqrt{p} - \sqrt{\bq})^2 < \tau$ and $(\sqrt{\bq} - \sqrt{\Ex[\bq]})^2 \leq \tau$, then \eqref{eq:good-interval} is contradicted. Therefore, if $(\sqrt{\bq} - \sqrt{\Ex[\bq]})^2 \leq \tau$ then $(\sqrt{p} - \sqrt{\bq})^2 \geq \tau$ and the test outputs $\Reject$ as desired.
    To conclude, we again apply \Cref{fact:Chernoff-hellinger} to bound the failure probability.
    \begin{equation*}
        \Prx\bracket*{\dhels(\Ber(\bestI), \Ber(\bp(I)))\geq \tau}\leq 2e^{-m\tau} \leq \delta. \qedhere
    \end{equation*}
    \end{description}
\end{proof}

\section{Proof of \Cref{thm:main-poisson}}
This section is dedicated to the proof of \Cref{thm:main-poisson}, which gives our nearly instance-optimal uniformity testing algorithm in the Poissonized setting.
\begin{theorem}[Formal version of \Cref{thm:main-poisson}]
    \label{thm:full-tester-formal}
    For any $\mu, \delta \geq 0$ and $n \in \N$, setting
    \begin{equation*}
        s \coloneqq O\paren*{(\log n)^2 \cdot \paren*{\log n + \log \mu + \log(1/\delta)}}
    \end{equation*}
    there is an efficient tester $T$ (\Cref{fig:full tester}) with the following guarantees.
    \begin{enumerate}
        \item[$\circ$] Given a sample from $\Poi(s\cdot \mu)^n$, $T$ outputs $\Accept$ with probability at least $1 - \delta$.
        \item[$\circ$] For any $\lambda_1, \ldots, \lambda_n \geq 0$ for which $\dtv(\Poi(\mu)^n, \Perm(\Poi(\lambda_1),\dots,\Poi(\lambda_n))) \geq 1/2$, given a sample from $\Perm(\Poi(s\cdot\lambda_1),\dots, \Poi(s\cdot\lambda_n))$, $T$ outputs $\Reject$ with probability at least $1 - \delta$.
    \end{enumerate}
\end{theorem}
The key novel ingredient in this section is the following.

\begin{lemma}[Distinguishing permutation distributions using subsamples, restatement of \Cref{lem:subsample-is-distinguishable} using Hellinger distance]
    \label{lem:subsample-is-distinguishable-hellinger}
    For any distributions $\bp_1, \ldots, \bp_n$ and $\bq$ for which $\Perm(\bp_1, \ldots, \bp_n)$ is distinguishable from $\bq^n$, there exist some $k \in [n]$ for which,
    \begin{equation*}
        \Ex_{\bi_1, \ldots, \bi_{k} \wor [n]}\bracket*{\dhels\paren*{\bq, \frac{\bp_{\bi_1} + \cdots + \bp_{\bi_k}}{k}}} \geq \Omega\paren*{\frac{1}{k \log n}}.
    \end{equation*}
    where $\wor$ denotes sampling uniformly without replacement.
\end{lemma}

% \begin{lemma}
%     \label{lem:subsample-is-distinguishable}
%     For any distributions $\bp_1, \ldots, \bp_n$ and $\unif$, suppose the total variation distance between $\unif^n$ and $\Perm(\bp_1, \ldots,\bp_n)$ is at least $1/2$. Then, there exist some $k \in [n]$ for which,
%     \begin{equation*}
%         \Ex_{\bi_1, \ldots, \bi_{k} \wor [n]}\bracket*{\dhels\paren*{\unif, \frac{\bp_{\bi_1} + \cdots + \bp_{\bi_k}}{k}}} \geq \Omega\paren*{\frac{1}{k \log n}}.
%     \end{equation*}
%     where $\wor$ denotes sampling uniformly without replacement.
% \end{lemma}
Note that \Cref{lem:subsample-is-distinguishable} follows from \Cref{lem:subsample-is-distinguishable-hellinger} by applying a ``reverse" Markov's inequality to the quantity inside the expectation: For any random variable $\bx$ upper bounded by $2$, like the squared Hellinger distance in \Cref{lem:subsample-is-distinguishable-hellinger}, $\Pr[\bx \geq \Ex[\bx]/2] \geq \Ex[\bx]/4$.

In the remainder of this section, we first show how to apply \Cref{lem:subsample-is-distinguishable-hellinger} to prove \Cref{thm:main-poisson} in \Cref{subsection:tester-and-analysis} and then prove \Cref{lem:subsample-is-distinguishable-hellinger} in \Cref{subsection:proof-of-subsample}.

\subsection{The nearly instance-optimal tester and its analysis}
\label{subsection:tester-and-analysis}

We will rely on the fact that, for any integer $s \geq 1$, a random variable drawn from $\Poi(s\lambda)$ can be split into $s$ i.i.d.\ random variables, each distributed as $\Poi(\lambda)$.
\begin{restatable}[Poisson Splitting]{fact}{PoissonSplit}
    \label{fact:poisson-split}
    For any $s \in \N$, there exists an efficient randomized algorithm $\Split_s\colon \N \to \N^s$ with the property that if $\by \sim \Poi(s\lambda)$, then the output of $\Split_s(\by)$ will be $s$ independent copies of $\Poi(\lambda)$.
\end{restatable}
\begin{proof}
   $\Split_s(y)$ draws $y$ many samples from the uniform distribution over $[s]$ and then returns $(\bx_1, \ldots, \bx_s)$ where $\bx_i$ is the number of times the element $i$ appears in that size-$y$ sample.  Correctness follows directly from \Cref{fact:poissonization}.
\end{proof}

%\noindent For completeness, we include a proof of \Cref{fact:poisson-split} in \Cref{appendix:poisson-split}. Here, we proceed to proving~\Cref{thm:full-tester-formal}.
\begin{proof}[Proof of~\cref{thm:full-tester-formal}]
We separate the correctness analysis of~\Cref{fig:full tester} into the two natural cases.
\algdef{SE}[REPEATN]{RepeatN}{End}[1]{\algorithmicrepeat\ #1 \textbf{times}}{\algorithmicend}
\begin{algorithm}[htb] 
    \begin{algorithmic}[1] 
        \Require Parameters $\mu, \tau, s, r, x_{\max}$ and samples $y^{(1)}, \ldots, y^{(n)} \in \N$.
        
        \ForAll{$i \in [n]$}
         \State split the samples, $\bx_1^{(i)}, \ldots, \bx_s^{(i)} \leftarrow \Split_s(y^{(i)})$ \Comment{$\Split_s$ is the procedure from \Cref{fact:poisson-split}}
        \EndFor
        \ForAll{ $k \in [n]$ and every interval $I = [a,b]$ where $a \leq b \in [x_{\max}]$}
            \RepeatN{$r$}
                \State Draw $\bi_1, \ldots, \bi_k \wor [n]$ and define the parameters
                \begin{equation*}
                    \muI \coloneqq \Prx_{\bx \sim \Poi(\mu)}[\bx \in I]\quad\quad \bestI\coloneqq \frac{1}{sk} \cdot \sum_{j \in [k], \ell \in [s]} \Ind[\bx^{(\bi_j)}_\ell \in I].
                \end{equation*}
                \If{ $\dhels(\Ber(\muI), \Ber(\bestI))\geq \tau/k$ } \Return $\Reject$ and halt.
                \EndIf
            \End
         \EndFor   
        \State \Return $\Accept$.
    \end{algorithmic}
\caption{Overall tester for distinguishing between samples from a Poisson with mean $\mu \cdot r$ and a permutation distribution of Poissons.}
\label{fig:full tester}
\end{algorithm}

% We separate the correctness of the tester in \Cref{fig:full tester} into the natural two cases.
\begin{lemma}
    \label{lem:full-tester-uniform}
    For $\by^{(1)}, \ldots, \by^{(n)} \iid \Poi(s \mu)$, the probability the tester in \Cref{fig:full tester} outputs $\Reject$ is at most $O(x_{\max}^2 \cdot n \cdot r \cdot e^{-s \tau})$.
\end{lemma}
\begin{proof}
    The tester checks $O(x_{\max}^2)$ intervals. For each interval, it tests $n$ choices of $k$ and repeats these tests $r$ times. Hence, by the union bound, it suffices to show that for any such interval and choice of $k$, the probability $\dhels(\Ber(\muI), \Ber(\bestI)) \geq \tau/k$ is at most $O(e^{-s \tau})$.

    Applying \Cref{fact:poisson-split} and using the fact that the samples $\by^{(1)}, \ldots, \by^{(n)}$ are i.i.d.\ from $\Poi(s \mu)$, we have that the random variables $\{\bx_{j}^{(i)}\}_{i \in [n], j \in [s]}$ are i.i.d.\ from $\Poi(\mu)$. Hence, for any interval $I$ and any choice of $i_1, \ldots, i_k$, the random variable $\bestI$ is the average of $sk$ binary-valued random variables each with mean $\muI$. Using the concentration inequality of \Cref{fact:Chernoff-hellinger},
    \begin{equation*}
        \Prx\bracket*{\dhels(\Ber(\muI), \Ber(\bestI)) \geq \tau/k} \leq 2e^{-sk \cdot \tau/k} = 2e^{-s\tau}\,.\qedhere
    \end{equation*}
\end{proof}

\begin{lemma}
    \label{lem:full-tester-nonuniform}
    For any $\lambda_1, \ldots, \lambda_n$ for which 
    \[
    \dtv(\Poi(\mu)^n, \Perm(\Poi(\lambda_1),\dots, \Poi(\lambda_n))) \geq \Omega(1)
    \]
    and $\by^{(1)}, \ldots, \by^{(n)} \sim \Perm(\Poi(\lambda_1),\dots, \Poi(\lambda_n))$, if $\tau \leq \Theta(1/ (\log n)^2)$ and $x_{\max} \geq \Theta(\mu + \log(n))$, the probability that the tester in \Cref{fig:full tester} outputs $\Accept$ is at most $e^{-\Omega(\tau s)} + e^{-\Omega(r/(n \log n))}$.
\end{lemma}
\begin{proof}
    By \Cref{lem:subsample-is-distinguishable-hellinger}, there is some $k \in [n]$ for which
    \begin{equation*}
         \Ex_{\bi_1, \ldots, \bi_{k} \wor [n]}\bracket*{\dhels\paren*{\Poi(\mu), \frac{\Poi(\lambda_{\bi_1}) + \cdots + \Poi(\lambda_{\bi_k})}{k}}} \geq \Omega\paren*{\frac{1}{k \log n}}.
    \end{equation*}
    Here, we apply the reverse Markov's inequality: For any random variable $\bx$ bounded above by $2$, such as the squared Hellinger distance in the above expression, $\Prx[\bx \geq \Ex[\bx]/2] \geq \Ex[\bx]/4$. Hence,
    \begin{equation*}
         \Prx_{\bi_1, \ldots, \bi_{k} \wor [n]}\bracket*{\dhels\paren*{\Poi(\mu), \frac{\Poi(\lambda_{\bi_1}) + \cdots + \Poi(\lambda_{\bi_k})}{k}} \geq  \Omega\paren*{\frac{1}{k \log n}}} \geq \Omega\paren*{\frac{1}{k \log n}} \geq \Omega\paren*{\frac{1}{n \log n}}.
    \end{equation*}
    Let us say that $i_1, \ldots, i_k$ are ``good" if, for $\bp \coloneqq \frac{\Poi(\lambda_{\bi_1}) + \cdots + \Poi(\lambda_{\bi_k})}{k}$, $\dhels\paren*{\Poi(\mu), \bp} \geq\Omega\paren*{\tfrac{1}{k \log n}}$. The tester in \Cref{fig:full tester} repeats the test $r$ times for each $k \in [n]$. Therefore, with probability at least $1 - e^{-\Omega(r/(n \log n))}$, the tester will evaluate some $i_1, \ldots, i_k$ that are ``good."

    For now, let us assume it finds such a ``good" $i_1, \ldots, i_k$ and let $\bp \coloneqq \frac{\Poi(\lambda_{\bi_1}) + \cdots + \Poi(\lambda_{\bi_k})}{k}$ denote the corresponding mixture distribution. If it does, by \Cref{lem:good-interval}, there is an interval $I$ the tester checks for which
    \begin{equation*}
       \dhels(\Ber(\Poi(\mu)(I)), \Ber(\bp(I))) \geq \Omega\paren*{\frac{1}{k \log n \log(k \log n)}} \geq \Omega\paren*{\frac{1}{k (\log n)^2}} \geq \frac{4\tau}{k}.
    \end{equation*}
    Now let $\bestI$ be the estimate that the tester in \Cref{fig:full tester} computes for this particular $i_1, \ldots, i_k$ and interval $I$. By \Cref{cor:almost-triangle-hellinger}, if $\dhels(\Ber(\bestI)), \Ber(\bp(I))\leq \tau/k$ then $\dhels(\Ber(\Poi(\mu)(I)), \Ber(\bestI)) \geq \tau/k$ and the test outputs $\Reject$. Therefore, for the test to fail to $\Reject$, it must be the case that $\dhels(\Ber(\bestI)), \Ber(\bp(I)) \geq \tau/k$.

    We bound the probability this occurs. By \Cref{fact:poisson-split}, the random variables $\{\bx_{j}^{(i)}\}_{i \in [n], j \in [s]}$ are all independent and each have the marginal distribution $\Poi(\lambda_i)$. Hence, $\bestI$ is the average of $sk$ many binary random variables and has mean $\bp(I)$. Using the bound of \Cref{fact:Chernoff-hellinger},
    \begin{equation*}
        \Pr[\dhels(\Ber(\bestI)), \Ber(\bp(I)) \geq \tau/k] \leq  2e^{-sk \cdot \tau/k} = 2e^{-s \tau}.
    \end{equation*}
    Hence, in order for the test to output $\Accept$, it must either never find a ``good" $i_1, \ldots, i_k$ or have $\dhels(\Ber(\bestI)), \Ber(\bp(I)) \geq \tau/k$ when it does find such a good $i_1, \ldots, i_k$. Union bounding over these two failure cases gives the desired bound. 
    % Fix the interval $I$ for which the above inequality holds. Then, defining $\muI \coloneqq \Poi(\mu)(I)$ and $\wh{\est} \coloneqq \sqrt{\mcM(\lambda_{i_1}, \ldots, \lambda_{i_k})(I)}$, and by the bound on $\tau$,
    % \begin{equation*}
    %     \paren*{\sqrt{q} - \sqrt{\hat{p}}}^2 \geq \Omega\paren*{\frac{1}{k (\log n)^2}} \geq \frac{\tau}{4k}.
    % \end{equation*}
    % Let $\bq$ be the quantity that the tester in \Cref{fig:full tester} computes. By the almost triangle inequality \Cref{fact:almost-triangle-inequality} and the above bound, we have that either $(\sqrt{\bq} - \sqrt{\hat{q}})^2 \geq \tau/k$ or $(\sqrt{p} - \sqrt{\bq})^2 \geq \tau/k$. In the later case, the tester outputs $\Reject$, so for it to output $\Accept$ it must be the case that $(\sqrt{\bq} - \sqrt{\hat{q}})^2$.

    % We bound the probability this occurs. By \Cref{fact:poisson-split}, the random variables $\{\bx_{j}^{(i)}\}_{i \in [n], j \in [s]}$ are all independent and each have the marginal distribution $\Poi(\lambda_i)$. Hence, $\bq$ is the average of $sk$ many binary random variables and has mean $\hat{q}$. Using the bound of \Cref{fact:Chernoff-hellinger},
    % \begin{equation*}
    %     \Pr[(\sqrt{\bq} - \hat{\bq})^2 \geq \tau/(4k)] \leq  2e^{-sk \cdot \tau/(4k)} = 2e^{-s \tau /4}.
    % \end{equation*}
    % Hence, in order for the test to output $\Accept$, it must either never find a ``good" $i_1, \ldots, i_k$ or have $(\sqrt{\bq} - \hat{\bq})^2 \geq \tau/(4k)$. Union bounding over these two failure cases gives the desired bound.
\end{proof}

All that remains is to complete the proof of \Cref{thm:full-tester-formal} is to show how to set the parameters of \Cref{fig:full tester}: choosing
    \begin{align*}
        x_{\max} &= \Theta(\mu + \log n) \\
        \tau &= \Theta\paren*{\lfrac{1}{(\log n)^2}} \\
        r &= \Theta\paren*{\log(1/\delta) \cdot  n \log n} \\
        s &= \Theta\paren*{(\log n)^2 \cdot \paren*{\log n + \log \mu + \log(1/\delta)}}\,,
    \end{align*}
    \Cref{thm:full-tester-formal} then follows from \Cref{lem:full-tester-uniform,lem:full-tester-nonuniform}.
\end{proof}

\subsection{Proof of \Cref{lem:subsample-is-distinguishable-hellinger}}
\label{subsection:proof-of-subsample}
To establish the last missing piece, \Cref{lem:subsample-is-distinguishable-hellinger}, we will use the following recent (and surprisingly non-obvious) result showing that squared Hellinger distance satisfies an approximate chain rule.
\begin{theorem}[Approximate chain rule for squared Hellinger distance \cite{FHQR24,jay09}]
    \label{fact:approx-chain-rule}
    Let $\bp, \bq$ be distributions over a domain $\mcX^n$ and let $\bp_i(x_{\leq i-1})$ (resp., $\bq_i(x_{\leq i-1})$) be the distribution of $\bx_i$ conditioned on $\bx_{\leq i-1} = x_{\leq i-1}$ where $\bx \sim \bp$ (resp., $\bx \sim \bq$). There is an absolute constant $c \leq 7$ for which
    \begin{equation*}
        \dhels(\bp, \bq) \leq c \cdot \sum_{i \in [n]} \Ex_{\bx \sim \bp}\bracket*{\dhels(\bp_i(\bx_{\leq (i-1)}), \bq_i(\bx_{\leq (i-1)})}.
    \end{equation*}
\end{theorem}

We will also use the following standard fact about squared Hellinger distance (which, for example, can be derived from the fact that $(p,q) \mapsto (\sqrt{p} - \sqrt{q})^2$ is jointly convex in $p$ and $q$).
\begin{fact}[Conditioning can only increase squared Hellinger distance]
    \label{fact:conditioning-increases}
     For any random variables $\bx, \by$ and $\bz$, let $\bp_x, \bp_y,$ and $\bp_z$ denote their respective marginal distributions and $\bp_{x \mid y}(y)$ denote the distribution of $\bx$ conditioned on $\by = y$. Then,
     \begin{equation*}
         \Ex_{\by \sim \bp_y}\bracket*{\dhels(\bp_{x \mid y}(\by), \bp_z)} \geq \dhels(\bp_x, \bp_z).
     \end{equation*}
\end{fact}
\noindent With these in hand, we are able to establish the lemma:
\begin{proof}[Proof of \Cref{lem:subsample-is-distinguishable-hellinger}]
    We start with the assumption that $\dtv(\Perm(\bp_1, \ldots, \bp_n), \bq^n) \geq \Omega(1)$ which, by~\eqref{eq:tv:hellinger}, implies that $\dhels(\Perm(\bp_1, \ldots, \bp_n), \bq^n) \geq \Omega(1)$. Next, we apply \Cref{fact:approx-chain-rule}, which gives that,
    \begin{equation*}
        \sum_{i \in [n]}\Ex_{\bx \sim \Perm(\bp_1, \ldots, \bp_n)}\bracket*{\dhels(\bp(\bx_i \mid \bx_{\leq i-1}), \bq)} \geq \Omega(1).
    \end{equation*}
    where $\bp(\bx_i \mid \bx_{\leq i-1})$ is the marginal distribution of the $i$\textsuperscript{th} coordinate of $\bx \sim \Perm(\bp_1, \ldots, \bp_n)$ conditioned on the first $i-1$ coordinates. 
    
    From the above, we have that there is some $k \in [n]$ for which
    \begin{equation*}
        \Ex_{\bx \sim \Perm(\bp_1, \ldots, \bp_n)}\bracket*{\dhels(\bp(\bx_k \mid \bx_{\leq k-1}), \bq)} \geq \Omega\paren*{\lfrac{1}{(n-k + 1) \log n}}.
    \end{equation*}

    Recall a sample of $\bx_1, \ldots, \bx_n \sim \Perm(\bp_1, \ldots, \bp_n)$ can be drawn by first setting $\bi_1, \ldots, \bi_n$ to a uniform permutation of $[n]$ and then drawing $\bx_j \sim \bp_{\bi_j}$ independently for each $j \in [n]$. The way in which conditioning on $\bx_{\leq k-1} = x_{\leq k-1}$ affects the distribution of $\bx_k$ is by shifting the distribution of $\bi_1, \ldots, \bi_n$ to no longer be uniform over all permutations of $[n]$. In particular, it affects the distribution of $\bi_1, \ldots, \bi_{k-1}$, which means that $\bi_k$ will no longer be uniform over $[n]$. Unfortunately, this conditioning is quite complicated, and difficult to get directly analyze: to address this, we apply \Cref{fact:conditioning-increases} to ``add even more conditioning'', and fully condition on the values of $\bi_1, \ldots, \bi_{k-1}$, leading to a quantity easier to analyze. Combined with the prior bound, this gives
    \begin{equation*}
        \Ex_{\bi_1, \ldots, \bi_{k-1} \wor [n]}\bracket*{\dhels\paren*{\Ex_{\bi_k \sim [n] \setminus \set{\bi_1, \ldots, \bi_{k-1}}}\bracket*{\bp_{\bi_k}}, \bq}} \geq \Omega\paren*{\lfrac{1}{(n-k + 1) \log n}}.
    \end{equation*}
    Finally, we can ``flip" this expression by setting $\bj_1, \ldots, \bj_{n-k + 1}$ to be the elements of $ [n] \setminus \set{\bi_1, \ldots, \bi_{k-1}}$ and conclude that
    \begin{equation*}
        \Ex_{\bj_1, \ldots, \bj_{n-k + 1} \wor [n]}\bracket*{\dhels\paren*{\frac{\bp_{\bj_1} + \cdots + \bp_{\bj_{n-k+1}}}{n-k + 1}, \bq}} \geq\Omega\paren*{\lfrac{1}{(n-k + 1) \log n}}. 
    \end{equation*}
    Reparameterizing with the change of indices $\ell \coloneqq n-k + 1$ yields the lemma's statement.
\end{proof}

\section{Reduction to \Cref{thm:main-poisson}}
\label{sec:poissonization-reduction}
We prove that the formal version of~\cref{thm:main-instanceoptimal} stated below follows from \Cref{thm:main-poisson}.
\begin{theorem}[Formal version of~\cref{thm:main-instanceoptimal}]
    \label{thm:uniformity-testing-formal}
    For any $n, m \in \N$ and $\delta > 0$, there is an efficient tester using $m' = O(m \cdot (\log m)^2 \cdot \log (m/\delta))$ i.i.d.\ samples from an unknown distribution $\bp$ supported on $[n]$ satisfying
    \begin{enumerate}
        \item \emph{\textbf{Completeness}}: If $\bp= \unif$, it outputs $\Accept$ with probability at least $1-\delta$.
        \item \emph{\textbf{Soundness}}: If $\bp \neq \unif$ and $\opt(\bp) \leq m$, it outputs $\Reject$ with probability at least $1-\delta$.
    \end{enumerate}
\end{theorem}
The proof of \Cref{thm:main-poisson} mostly follows from the standard ``Poissonization" trick and \Cref{thm:full-tester-formal}. However, a direct application would have an overhead with both $\log m$ and $\log n$ terms. To eliminate the $\log n$ terms, we observe that when $m \ll \sqrt{n}$, the optimal tester essentially just rejects if it finds any collision.
\begin{claim}[Small $\opt$ implies high collision probability]
    \label{claim:small-opt-collision}
    For any distribution $\bp$ on $[n]$ with $\opt(\bp) = m \leq \sqrt{n}/2$ if we sample $\bx_1, \ldots, \bx_m \iid \bp$, then with probability at least $2/3$ there is some $i \neq j \in [m]$ for which $\bx_i \neq \bx_j$. 
\end{claim}
\begin{proof}
    We show that if the probability of a collision in $m$ samples is less than $2/3$, then $\opt(\bp) > m$. This implies the desired result by contrapositive. To do so, we consider any tester $T$ using $m$ samples and show that either $T$ does not accept the uniform distribution with probability $9/10$ or that it does not reject some relabeling of $\bp$ with probability $9/10$.

    Let $\alpha$ be the probability that $T$ accepts given $\bx_1, \ldots, \bx_m$ chosen uniformly from $[n]$ conditioned on all $m$ values being unique. We separate this analysis into two cases based on the value of $\alpha$.
    \begin{enumerate}
        \item[] \textbf{Case 1:} $\alpha \leq 1/2$. In this case, we will argue that $T$ fails to accept given samples from the uniform distribution with probability at least $0.9$. Using linearity of expectation, if given $\bx_1, \ldots, \bx_n \iid \bu$, the probability that all are unique is at least $1 - \frac{m^2}{n} \geq 3/4$. On such samples, $T$ rejects with probability at least $1/2$. Hence, given samples from the uniform distribution, it rejects with probability at least $3/4 \cdot 1/2 = 3/8$.
        \item[] \textbf{Case 2:} $\alpha \geq 1/2$. In this case, we will argue that $T$ fails to reject given samples from some relabeling of $\bp$ with probability at least $0.9$. Suppose we first draw a \emph{uniform} relabeling of $\bp$, denoted $\tilde{\bp}$, and then draw $\bx_1, \ldots, \bx_m \iid \tilde{\bp}$. Then, by assumption, the probability that all of $\bx_1, \ldots, \bx_m$ are unique is at least $1/3$. Furthermore, when they are unique, they are equally likely to take on any $m$ distinct values (since $\tilde{\bp}$ is a uniform relabeling).

        Then, since $\alpha \geq 1/2$, we can conclude that given samples from $\tilde{\bp}$, $T$ accepts with probability at least $1/2 \cdot 1/3 = 1/6$. In particular, this implies there exists a single relabeling of $\bp$ on which $T$ accepts with probability at least $1/6$, as desired.
    \end{enumerate}
    Hence, we have shown that any tester using $m$ samples cannot distinguish $\bu$ from all relabelings of $\bp$ unless the probability of a collision in samples from $\bp$ is at least $1/2$. The desired result follows by contrapositive.
\end{proof}
For the Poissonization case, we will invoke the following standard fact.
\begin{fact}
    \label{fact:poissonization}
    For any distribution $\bp$, if we take $\boldm \sim \Poi(m)$ and then draw $\boldm$ independent samples from $\bp$, then using $\bx_i$ to denote the number of times $i$ appears in the sample, $(\bx_1, \ldots, \bx_n)$ is distributed as the product distribution $\Poi(m\cdot\bp(1)) \times \cdots \times \Poi(m \cdot \bp(n))$. Furthermore, this process can be inverted: There is an algorithm that, without knowing $\bp$, takes as input $\bx \sim \Poi(m\cdot\bp(1)) \times \cdots \times \Poi(m \cdot \bp(n))$ and outputs $\Poi(m)$ independent samples from $\bp$.
\end{fact}
\begin{claim}
    \label{claim:opt-poissonization}
    If $\opt(\bp) \leq m$, then for $m' = O(m)$, $\mu \coloneqq m'/n$ and $\lambda_i \coloneqq \bp(i) \cdot m'$ for each $i \in [n]$,
    \begin{equation*}
        \dtv(\Poi(\mu)^n, \Perm(\Poi(\lambda_1), \ldots, \Poi(\lambda_n)) \geq 1/2.
    \end{equation*}
\end{claim}
\begin{proof}
    Let $T$ be the test using $m$ samples guaranteed to exist from the fact that $\opt(\bp) \leq m$. Using it, we will construct a test $T': \N^n \to \zo$ for which
    \begin{enumerate}
        \item Given a sample from $\Poi(\mu)^n$, $T'$ will accept with probability at least $3/4$.
        \item Given a sample from $\Perm(\Poi(\lambda_1), \ldots, \Poi(\lambda_n))$, $T'$ will reject with probability at least $3/4$.
    \end{enumerate}
    Such a guarantee implies the desired TV distance bound. This test $T'(x_1, \ldots, x_n)$ works as follows: First, it uses the ``inversion" part of \Cref{fact:poissonization} to generate $\bz \sim \Poi(m')$ many independent samples from some distribution $\bq$. If $\bx \sim \Poi(\mu)^n$, then it will generate samples from $\bq = \bu$. On the other hand, if $\bx \sim \Poi(\lambda_1), \ldots, \Poi(\lambda_n)$, then it will have generated samples from a uniform relabeling of $\bp$.

    Then, if the number of samples it has generated, $\bz$, it at least $m$, it accepts iff $T$ accepts on the first $m$ samples. on the other hand, if $\bz < m$, it simply defaults to rejecting. By the guarantees of $T$ and \Cref{fact:poissonization}, the failure probability of this test $T'$ is upper bounded by the sum of the failure probability of $T$ (which is $0.1$) and the probability $\bz < m$. If we set $m' \coloneqq \max(2m, 20)$, then $\Pr_{\bz \sim \Poi(m')}[\bz < m] \leq 0.1$, giving a good enough bound on the total failure probability.
\end{proof}

Finally, we prove our main result for uniformity testing.
\begin{proof}[Proof of \Cref{thm:uniformity-testing-formal}]
    First, we handle the case where $m \leq \sqrt{n}/2$ using \Cref{claim:small-opt-collision}. In this case, our tester will use $m' \coloneqq O(m \log(1/\delta))$ samples. It breaks these into $O(\log 1/\delta))$ groups each containing $m$ samples and rejects if and only if a majority of groups contain at least one collision.

    By \Cref{claim:small-opt-collision}, if $\opt(\bp) \leq m$, then the probability each group has a collision is at least $2/3$. On the other hand, if $\bp  = \bu$, the probability of a collision in each group is at most $m^2/n \leq 1/4$. By a standard concentration inequality (e.g., \Cref{fact:Chernoff-KL}), we can conclude that the failure probability of this tester is at most $\delta$ in both cases.

    Next, we handle the case where $m > \sqrt{n}/2$. In this case we use Poissonization and \Cref{thm:full-tester-formal}. By \Cref{claim:opt-poissonization} we know that for $m' = O(m)$, $\mu \coloneq m'/n$, and $\lambda_i \coloneqq \bp(i) \cdot m'$ that $\Poi(\mu)^n$ and $\Perm(\Poi(\lambda_1), \ldots, \Poi(\lambda_n)$ are distinguishable. Hence for
    \begin{equation*}
        s \coloneqq O((\log n)^2 \cdot (\log n + \log \mu + \log(1/\delta))) = O((\log m)^2 \log(m/\delta)),
    \end{equation*}
    the tester of \Cref{thm:full-tester-formal} distinguishes between $\Poi(s \cdot \mu)^n$ and $\Perm(\Poi(s\cdot\lambda_1),\dots, \Poi(s\cdot\lambda_n))$ with probability at least $1 - \delta$. By \Cref{fact:poissonization}, if we draw $\bz \sim \Poi(O(sm))$ and then take $\bz$ samples from $\bu$, the frequency vector will be distributed as $\Poi(s \cdot \mu)^n$. Similarly, if we draw $\bz$ samples from $\bp$, the frequency vector will be distributed as $\Poi(s \cdot \lambda_1) \times \cdots \Poi(s \cdot \lambda_n))$, which we can uniformly permute to obtain a sample of $\Perm(\Poi(s\cdot\lambda_1),\dots, \Poi(s\cdot\lambda_n))$.

    Hence, with $\bz$ samples, we can use the tester \Cref{thm:full-tester-formal} to distinguish the two cases. Here, we use \Cref{fact:Poisson-tail} to upper bound $\bz$: With probability at least $1-\delta$, $\bz$ is at most $O(sm + \log(1/\delta))$. Note that $s \geq \log(1/\delta)$, so we only need $O(sm) =  O(m \cdot (\log m)^2 \cdot \log (m/\delta))$ samples to have enough samples with probability at least $1-\delta$.

    This tester can only fail if (1) it does not have enough samples or (2) the tester from \Cref{thm:full-tester-formal} fails. Union bounding over these two failure probabilities gives a tester that fails with probability $2\delta$. We simply reparameterize $\delta = \delta'/2$ to obtain the desired result.
\end{proof}

% \clearpage
% \bgroup\color{gray}
%  \input{Notes}
% \egroup

\section*{Acknowledgements}

We thank the FOCS reviewers for helpful feedback and suggestions. 

Guy is supported by a Jane Street Graduate Research Fellowship, NSF awards 1942123, 2211237, 2224246, a Sloan Research Fellowship, and a Google Research Scholar Award. Cl\'ement is supported by an ARC DECRA (DE230101329) from the Australian Research Council. Erik Waingarten is supported by the National Science Foundation (NSF) under
Grant No. CCF-2337993.

\bibliographystyle{alpha}
\bibliography{ref}

\appendix
\section{Deferred proofs}
\label{appendix:deferred}
\subsection{Proof of~\Cref{lem:guess}}
\label{appendix:proof-of-guess}
\begin{proof}
Given an algorithm $T$ as per the statement of the lemma, the new tracking algorithm $T'$ proceeds as follows. Starting with $m = 1$, it repeats the following loop:
\begin{itemize}
\item Instantiate an independent instance of the algorithm $T$ with parameter $m$ while outputting \Plausible;
\item When the algorithm terminates, output \Reject and terminate if the algorithm outputs \Reject. If the algorithm outputs `\Accept, output \Plausible and update $m \gets 2m$.
\end{itemize}
If $\bp$ is uniform, then, writing $m = 2^h$ for some $h \in \Z_{\geq 0}$, we have
\begin{align*}
\Prx\left[ T' \text{ outputs } \Reject\right] &\leq \sum_{h=0}^{\infty} \Prx\left[ T\text{ with }m = 2^h\text{ outputs }\Reject\right] \leq \sum_{h=0}^{\infty} \delta / (2 \cdot 2^h) \leq \delta. 
\end{align*}
If $\bp$ is not uniform, then let $h \in \Z_{\geq 0}$ be the smallest satisfying $\opt(\bp) \leq 2^h$. For every $h' \geq h$, the probability that the $T$ ran with $h'$ does not terminate is at most $1/10$, and in particular, the probability that the new algorithm $T'$ does not terminate by $h'$ is at most $(1/10)^{h'-h}$. Thus, the expected sample complexity of our algorithm is at most
\begin{align*}
\sum_{\ell=0}^{h-1} s(2^{\ell}) + \sum_{h' > h} \left( \frac{1}{10}\right)^{h' - h} \cdot s(2^{h'}) \leq 2c \cdot 2^{h} + c \cdot 2^h \sum_{\ell=1}^{\infty} \left( \frac{2}{10}\right)^{\ell} \leq O(c) \cdot 2^{h} \leq O(c) \cdot \opt(\bp)\,,
\end{align*} 
proving the lemma.
\end{proof}

\subsection{Proof of \Cref{prop:eliminate-large}}
\label{appendix:proof-of-eliminate-large}
We prove the following restated for convenience.
\eliminateLarge*

This proof will use the following standard generalization of the almost triangle inequality \Cref{fact:almost-triangle-inequality}.
\begin{proposition}[Generalization of the almost triangle inequality]
    \label{prop:almost-triangle-gen}
    For any $x_1, \ldots, x_n$,
    \begin{equation*}
        \paren*{x_1 + \cdots + x_n}^2 \leq n\cdot \paren*{x_1^2 + \cdots + x_n^2}.
    \end{equation*}
\end{proposition}
\begin{proof}
    This is a direct application of Cauchy–Schwarz, taking $u = [x_1, \ldots, x_n]$ and $v = [1,\ldots, 1]$, the left-hand side is $(u \cdot v)^2$ and the right-hand side is $\ltwo{u}^2 \ltwo{v}^2.$
\end{proof}

Using the above, we prove the following:
\begin{proposition}
    \label{prop:almost-triangle-hel}
    For any $p,q, p', q'\geq 0$,
    \begin{equation*}
        \paren*{\sqrt{p} - \sqrt{q}}^2 \leq 3\paren*{\paren*{\sqrt{p'} - \sqrt{q'}}^2 +\abs*{p'-p} + \abs*{q'-q}}.
    \end{equation*}
\end{proposition}
\begin{proof}
    Since $\sqrt{p} - \sqrt{q} = (\sqrt{p'} - \sqrt{q'}) + (\sqrt{p} - \sqrt{p'}) + (\sqrt{q'} - \sqrt{q})$ we can apply \Cref{prop:almost-triangle-gen} giving
    \begin{equation}
        \label{eq:expand-hel-diff}
        \paren*{\sqrt{p} - \sqrt{q}}^2  \leq 3\paren*{\paren*{\sqrt{p'} - \sqrt{q'}}^2 + \paren*{\sqrt{p} - \sqrt{p'}}^2 + \paren*{\sqrt{q} - \sqrt{q'}}^2 }.
    \end{equation}
    We then bound,
    \begin{align*}
        \paren*{\sqrt{p} - \sqrt{p'}}^2 &= \frac{\paren*{\sqrt{p} - \sqrt{p'}} \cdot \paren*{\sqrt{p} - \sqrt{p'}} \cdot \paren*{\sqrt{p} +\sqrt{p'}}}{\paren*{\sqrt{p} +\sqrt{p'}}} \\
        &= (p - p') \cdot \frac{\sqrt{p} - \sqrt{p'}}{\sqrt{p} +\sqrt{p'}}
    \end{align*}
    The magnitude of the second term is upper bounded by $1$ (which occurs when  exactly one of $p' = 0$ or $p = 0$). Hence $(\sqrt{p} - \sqrt{p'})^2 \leq \abs*{p' - p}$. Using this bound for the two rightmost terms of \Cref{eq:expand-hel-diff} gives the desired bound.
\end{proof}

We now prove the main result of this subsection.
\begin{proof}[Proof of \Cref{prop:eliminate-large}]
    We split this proof into two cases. If $\bp(T) \geq \delta/10$, then we claim the desired result holds for $S' \coloneqq \overline{T}$. We expand
    \begin{equation*}
        \dhels(\Ber(\bp(S')), \Ber(\bq(S')) = \paren*{\sqrt{\bp(S')} - \sqrt{\bq(S')}}^2 +  \paren*{\sqrt{\bp(T)} - \sqrt{\bq(T)}}^2 \geq  \paren*{\sqrt{\bp(T)} - \sqrt{\bq(T)}}^2  
    \end{equation*}
    With the constraints that $\bp(T) \geq \delta/10$ and $\bq(T) \leq \delta/20$, the quantity $\paren*{\sqrt{\bp(T)} - \sqrt{\bq(T)}}^2  $ is minimized exactly when these constraints are tight, in which case it is $(\sqrt{\delta/10} - \sqrt{\delta/20})^2 = \delta\paren*{\sqrt{1/10} - \sqrt{1/20}}^2 \geq \delta/120$, as desired.

    Otherwise, we have that $\bp(T) \leq \delta/10$. In this case, we will set $S' = S \setminus T$ with the intuition is we can apply \Cref{prop:almost-triangle-hel} to show that that subtracting $T$ does not affect the Hellinger distance much. Formally, let
    \begin{equation*}
        %\label{eq:def-params}
        p \coloneqq \bp(S),\quad\quad p' \coloneqq \bp(S'), \quad\quad q\coloneqq \bq(S),\quad\quad\text{and}\quad q'\coloneqq \bq(S').
    \end{equation*}
    Then,
    \begin{equation*}
        \dhels(\Ber(p), \Ber(q) = \paren*{\sqrt{p} - \sqrt{q}}^2 + \paren*{\sqrt{1-p} - \sqrt{1-q}}^2 
    \end{equation*}
    Applying \Cref{prop:almost-triangle-hel} to each of the above two terms, we have that
    \begin{align*}
        \dhels(\Ber(p), \Ber(q)) &\leq  3\paren*{\paren*{\sqrt{p'} - \sqrt{q'}}^2 +\abs*{p'-p} + \abs*{q'-q}} \\
        &\quad+  3\paren*{\paren*{\sqrt{1-p'} - \sqrt{1-q'}}^2 +\abs*{p'-p} + \abs*{q'-q}} \\
        &= 3 \cdot \dhels(\Ber(p'), \Ber(q')) + 6\abs*{p - p'} + 6\abs*{q-q'}.\\
        &\leq 3 \cdot \dhels(\Ber(p'), \Ber(q')) + 6(\delta/10 + \delta/20) \tag{$\bp(T) \leq \delta/10$ and $\bq(T) \leq \delta/20$ }.
    \end{align*}
    Rearranging and substituting back the original bound that $\dhels(\Ber(p), \Ber(q)) \geq \delta$ we have that
    \begin{equation*}
        \dhels(\Ber(\bp(S')), \Ber(\bq(S'))) \geq \frac{\delta - 6(\delta/10+ \delta/20)}{3} = \delta/30.
    \end{equation*}
    This also satisfies the desired bound.
\end{proof}

\end{document}